%% file: main.tex
\documentclass[journal]{IEEEtran}
\usepackage{amsmath,amssymb,amsthm}
\usepackage{graphicx}
\usepackage{epstopdf}
\usepackage{cite}
\usepackage{algorithm}
\usepackage{algpseudocode}

\usepackage{subfigure}
\usepackage{tikz}
\usepackage{pgfplots}
\usepgfplotslibrary{fillbetween}
\usepgfplotslibrary{groupplots}

\IEEEoverridecommandlockouts
 
\theoremstyle{plain}
\newtheorem{theorem}{Theorem}

\newtheorem{lemma}{Lemma}
\newtheorem{proposition}{Proposition}

\newtheorem{assumption}{Assumption}
\theoremstyle{definition}
\newtheorem{definition}{Definition}

\theoremstyle{remark}
\newtheorem{remark}{Remark}

\DeclareMathOperator{\spn}{span}

\DeclareMathOperator{\tr}{tr}
\DeclareMathOperator{\diag}{diag}

\DeclareMathOperator*{\argmin}{arg\,min}
\DeclareMathOperator*{\argmax}{arg\,max}
\DeclareMathOperator{\E}{\mathbb{E}}
\newcommand{\R}{\mathbb{R}}
\newcommand{\real}{\mathbb{R}}

\newcommand{\calS}{\mathcal S}
\newcommand{\calF}{\mathcal F}
\newcommand{\calA}{\mathcal A}
\newcommand{\calH}{\mathcal H}
\newcommand{\calC}{\mathcal C}
\newcommand{\calP}{\mathcal P}

\newcommand{\calR}{\mathcal R}

\newcommand{\EX}{\mathbb{E}}
\newcommand{\Kappa}{K}

\newcommand{\bD}{\boldsymbol{D}}
\newcommand{\bDi}{\bD_{\textbf{h}}}

\begin{document}

\title{
Multi-task Reinforcement Learning in Reproducing Kernel Hilbert Spaces via Cross-learning
}
\author{Juan Cervi\~no, Juan Andr\'es Bazerque, Miguel Calvo-Fullana, and Alejandro Ribeiro.
\thanks{
This work is supported by ARL DCIST CRA W911NF-17-2-0181 and the Intel Science and Technology Center for Wireless Autonomous Systems.
}
\thanks{
J. Cervi\~no, M. Calvo-Fullana, and A. Ribeiro are with the Department of Electrical and Systems Engineering, University of Pennsylvania, Philadelphia, PA 19104, USA (e-mail: \mbox{jcervino}@seas.upenn.edu; \mbox{cfullana}@seas.upenn.edu; \mbox{aribeiro}@seas.upenn.edu).
}
\thanks{
J. A. Bazerque is with the Department of Electrical Engineering, School of Engineering, UdelaR, Uruguay (e-mail: \mbox{jbazerque}@fing.edu.uy)  
}
\thanks{
This work has been presented in part at the 2019 American Control Conference (ACC)\cite{cervino2019meta}.
}
}

\maketitle

\begin{abstract}
Reinforcement learning (RL) is a framework to optimize a control policy using rewards that are revealed by the system as a response to a control action. In its standard form, RL involves a single agent that uses its policy to accomplish a specific task.~These methods require large amounts of reward samples to achieve good performance, and may not generalize well when the task is modified, even if the new task is related.~In this paper we are interested in a collaborative scheme in which multiple agents with different tasks optimize their policies jointly. To this end, we introduce \textit{cross-learning}, in which agents tackling related tasks have their policies constrained to be close to one another.~Two properties make our new approach attractive: (i) it produces a \textit{multi-task} central policy that can be used as a starting point to adapt quickly to one of the tasks trained for, in a situation when the agent does not know which task is currently facing, and (ii) as in \textit{meta-learning}, it adapts to environments related but different to those seen during training. We focus on continuous policies belonging to reproducing kernel Hilbert spaces for which we bound the distance between the task-specific policies and the cross-learned policy.~To solve the resulting optimization problem, we resort to a projected policy gradient algorithm and prove that it converges to a near-optimal solution with high probability.~We evaluate our methodology with a navigation example in which agents can move through environments with obstacles of multiple shapes and avoid obstacles not trained for.
\end{abstract}


\IEEEpeerreviewmaketitle

\input{01_introduction.tex}
\input{02_problem_formulation.tex}
\input{03_reinforcement_learning.tex}

\input{04_geometric_interpretation.tex}

\input{05_algorithm.tex}

\input{06_convergence_analysis.tex}
\input{07_numerical_results.tex}
\input{08_conclusions.tex}

\appendices
\input{09_appendix.tex}

\bibliographystyle{IEEEtran}
\bibliography{bib}

\begin{IEEEbiography}[{\includegraphics[width=1in,height=1.25in,clip,keepaspectratio]{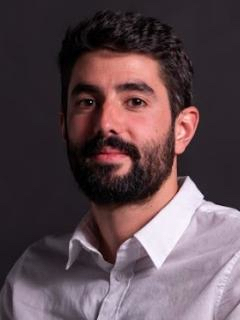}}]
  {Juan Cervi\~no} received the B.Sc. degree in electrical engineering from the Universidad de la Rep\'ublica Oriental del Uruguay, Montevideo, in 2018. He is a PhD student in the Department of Electrical and Systems Engineering at the University of Pennsylvania since 2019. Juan's current research interests are in machine learning, optimization and control.
\end{IEEEbiography}
\begin{IEEEbiography}[{\includegraphics[width=1in,height=1.25in,clip,keepaspectratio]{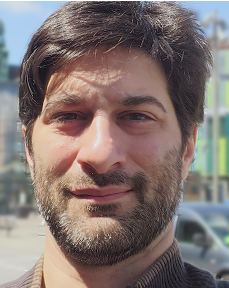}}]{Juan Andr\'es Bazerque} received the B.Sc. degree in electrical engineering from Universidad de la Rep\'ublica (UdelaR), Montevideo, Uruguay, in 2003, and the M.Sc. and Ph.D. degrees from the Department of Electrical and Computer Engineering, University of Minnesota (UofM), Minneapolis, in 2010 and 1013 respectively.
 Since 2015 he is an Assistant Professor with the  Department of Electrical Engineering  at UdelaR. His current research interests include  stochastic optimization and  networked systems, focusing on reinforcement learning, graph signal processing, and power systems optimization and control. 
Dr. Bazerque is the recipient of the UofM's Master Thesis Award 2009-2010, and co-reciepient of the best paper award at the 2nd International Conference on Cognitive Radio Oriented Wireless Networks and Communication 2007.   
 \end{IEEEbiography}
\begin{IEEEbiography}[{\includegraphics[width=1in,height=1.25in,clip,keepaspectratio]{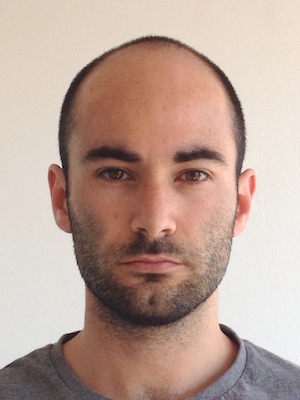}}]
  {Miguel Calvo-Fullana} received his B.Sc. degree in electrical engineering from the Universitat de les Illes Balears (UIB), in 2010 and the M.Sc. and Ph.D. degrees in electrical engineering from the Universitat Polit\`ecnica de Catalunya (UPC), in 2013 and 2017, respectively. From September 2012 to July 2013 he was a research assistant with Nokia Siemens Networks (NSN) and Aalborg University (AAU). From December 2013 to July 2017, he was with the Centre Tecnol\`ogic de Telecomunicacions de Catalunya (CTTC) as a research assistant. Since September 2017, he is a postdoctoral researcher at the University of Pennsylvania. His research interests lie in the broad areas of learning and optimization for autonomous systems. In particular, he is interested in multi-robot systems with an emphasis on wireless communication and network connectivity. He is the co-recipient of best paper awards at GlobalSIP 2015, ICC 2015, and ICASSP 2020.  
\end{IEEEbiography}
\begin{IEEEbiography}[{\includegraphics[width=1in,height=1.25in,clip,keepaspectratio]{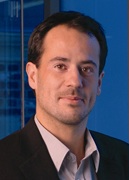}}]
    {Alejandro Ribeiro} received the B.Sc. degree in electrical engineering from the Universidad de la Republica Oriental del Uruguay, Montevideo, in 1998 and the M.Sc. and Ph.D. degree in electrical engineering from the Department of Electrical and Computer Engineering, the University of Minnesota, Minneapolis in 2005 and 2007. From 1998 to 2003, he was a member of the technical staff at Bellsouth Montevideo. After his M.Sc. and Ph.D studies, in 2008 he joined the University of Pennsylvania (Penn), Philadelphia, where he is currently Professor of Electrical and Systems Engineering. His research interests are in the applications of statistical signal processing to collaborative intelligent systems. His specific interests are in wireless autonomous networks, machine learning on network data and distributed collaborative learning. Papers coauthored by Dr. Ribeiro received the 2014 O. Hugo Schuck best paper award, and paper awards at CDC 2017, SSP Workshop 2016, SAM Workshop 2016, Asilomar SSC Conference 2015, ACC 2013, ICASSP 2006, and ICASSP 2005. His teaching has been recognized with the 2017 Lindback award for distinguished teaching and the 2012 S. Reid Warren, Jr. Award presented by Penn’s undergraduate student body for outstanding teaching. Dr. Ribeiro is a Fulbright scholar class of 2003 and a Penn Fellow class of 2015.    
\end{IEEEbiography}

\end{document}

%% file: 01_introduction.tex
\section{Introduction}
\label{sec:Introduction}

Reinforcement Learning (RL) provides a framework for solving a wide variety of problems in which the system model is not available or it may be intractable \cite{sutton2018reinforcement}. The idea at the core of RL is that an agent can learn the optimal action to take from its experience, gained while interacting with the environment. Specifically, the agent is aware of its current state and chooses an action, after which the next state is reached. Subsequently, after each action is taken, the agent receives a reward from the environment, which regulates the value of the actions chosen. The agent's objective is to maximize the cumulative reward observed throughout the trajectory of states and actions. Being able to tackle complex systems, RL has found success in applications as diverse as autonomous driving\cite{michels2005high,sallab2017deep}, robotics \cite{kober2013reinforcement}, and smart grids\cite{peters2013reinforcement}, among many others.

Albeit showing success in tackling complex problems, in order to do so RL  requires a large number of samples collected from actions and state transitions. Besides, a policy that achieves good performance may work poorly at a different task, even if it is related or akin to the one trained for. The attempt to alleviate these issues is the goal of Multi-Task Learning (MTL) \cite{caruana1997multitask}, either seeking to leverage information in related tasks, or learning multiple tasks in a collective manner. In a general sense, these approaches exploit the fact that the data collected when learning one task can help learning other similar tasks. There are numerous ways of posing MTL problems. When studied for supervised learning, one could upper bound the distance between the weight vectors of different linear classifiers \cite{kato2008multi,kato2009conic}. In the same vein, one could  use  convex local proximity constraints \cite{koppel2017proximity,koppel2019parsimonious}. Other approaches attempt to regulate which task should be sampled by maintaining an estimate of each current performance\cite{sharma2017learning}. Recently, a related approach named \textit{meta-learning} \cite{finn2017model}, has gained notorious popularity. The twist that makes it different from MTL is that the objective of meta-learning is to find a policy with good generalization with respect to unforeseen tasks. For the interested reader, a more exhaustive survey on multi-task learning can be found in \cite{zhang2017survey}.

In this work, we introduce a collaborative form of multi-task learning and meta-learning, which we denote as \emph{cross-learning}. We introduce convex constraints in the policy space that bound the distance between task specific policies corresponding to different agents. Compared to previous approaches \cite{kato2008multi,kato2009conic}, our work does not assume the relationships among tasks to be known. Furthermore, different from \cite{koppel2017proximity,koppel2019parsimonious}, we avoid pair-wise constraints via the introduction of a central \emph{cross-learned} policy. By nature of our proposed formulation, the benefits are twofold, (i) it produces a cross-learned policy that performs well for all the trained tasks, (ii) it produces an individual policy for each task that improves on an equivalent individual policy trained without cross-learning.

We focus on policies which lie on a reproducing kernel Hilbert space, allowing for the approximation of a large class of functions \cite{scholkopf2002learning}, and pose cross-learning  as a constrained RL problem. Then we show that the projection onto the resulting set of constraints can be cast as a Quadratically Constrained Quadratic Program (QCQP) \cite{boyd2009convex}, and propose a projected policy gradient solution. 
As an alternative, we offer a simplified formulation in which the coupling is relaxed to be satisfied on average by the policies. The simplified projection admits a closed form solution, both reducing the numerical burden of solving a  QCQP  per iteration, and clarifying the effect of the projection over the policies.     
We show the convergence of the projected gradient ascent method to a neighborhood of the optimal solution with high probability. Further, to handle the memory explosion due to the increasing dimension of the kernel representations, we use \cite{vincent2002kernel} and \cite{koppel2019parsimonious} to project the policies to a lower dimensional space without sacrificing convergence guarantees. Finally we evaluate the proposed cross-learning approach in a navigation problem with multiple obstacles. We show that our cross-learning method produces policies that adapt to different scenarios outperforming their task-specific counterparts, particularly when facing objects not seen during training.

The rest of the paper is structured as follows: Section \ref{sec:ProblemFormulation} introduces  RL  and its multi-task formulation. A treatment on how to solve the RL problem by means of gradient ascent is covered in Section \ref{sec:RL}. The aspects related to the necessary projection step in the cross-learning algorithm are covered in Section \ref{sec:GeoInter}, with Section \ref{subsec:relaxed} dedicated to its relaxed counterpart. The resulting cross-learning algorithm is covered in Section \ref{sec:CLRLA}, with practical issues regarding policy representation being discussed in Section \ref{subsec:sparse_projection}. Following, Section \ref{sec:convergence}  provides a detailed analysis of  convergence,  with Section \ref{subsec:asumptions} covering the assumptions taken to derive the results. The method is evaluated by means of simulations in Section \ref{sec:NumericalResults}, where a navigation problem is considered with single (Section \ref{subsec:NavigationSingle}) and multiple (Section \ref{subsec:NavigationMultiple}) obstacles. Finally, Section \ref{sec:Conclusions} closes the paper with some remarks and conclusions.

%% file: 02_problem_formulation.tex
\section{Problem Formulation}
\label{sec:ProblemFormulation}

Before introducing our cross-learning formulation, let  us present the Markov decision process that models a single-task scenario. Consider two compact sets $\calS \subset \R^q$ and $\calA \subset \R^p$, representing the possible states of an agent and the actions it can take. The agent takes actions sequentially. At time $t$, the agent is in state $s_t \in \calS$ and it takes an action $a_t \in \calA$, transitioning to a state $s_{t+1} \in \calS$. The aforementioned dynamics are governed by transition probabilities, unknown to the agent, which satisfy the following Markov property, $p(s_{t+1}=s | \{s_l, a_l\}_{l\leq t})=p(s_{t+1}=s | s_t, a_t)$. Furthermore, each time that the agent takes an action and transitions into the next state, it receives a reward defined by a function $r : \calS \times \calA \to \R$. 

In order to decide which actions to take, the agent has a function $h:\calS \to \calA$, called a policy, which in this work we assume to belong to a  reproducing kernel Hilbert space (RKHS) $\calH$. By definition, $\calH$ is a complete, linear function space endowed with a unique reproducing kernel $\bar \kappa:\calS \times \calS \to \R$. We further define $\kappa(s,s')$ as a diagonal matrix-valued function of size $p\times p$ where $\kappa(s,s')_{ii}=\bar \kappa(s,s')$.  Explicitly, every function $h \in \calH$ can be written as a linear combination of $M$ kernels (including $M=\infty$), i.e.,
\begin{align}
\label{eq:kernel_expansion}
h(\cdot)=\sum_{m=1}^{M}  \kappa(s_{m},\cdot) w_{m} 
\end{align}
where $w_{m} \in \calA$ denote the weights and $s_m \in \calS$ are called the kernel centers, or knots. Hence, every function $h$ can be represented by a dictionary $\textbf{D}_h = \{ s_1,\dots,s_M \}$ and a set of weights $\textbf{W}_h=\{w_1,\dots,w_M\}$. 

\begin{remark}
For $p=1$,   $\calH=\overline{\spn}\{\kappa(s,s'):\ s\in \calS\}$ is defined as a space of functions spanned by $\kappa(s,s')$ where $\kappa:\calS\times \calS\to \R $ consists of a symmetric positive definite kernel,  and equipped with the inner product defined by the reproducing property $\langle \kappa(s,\cdot),h(\cdot)\rangle=h(s)$. This inner product induces the norm $\|h\|=(\langle h,h\rangle)^{1/2}$ to be used henceforth. For $p>1$, $\kappa(s,s')\in\R^{p\times p}$ is a diagonal matrix whose entries satisfy the definition for $p=1$, and the inner product between two vector valued functions $h$ and $g$ aggregates the $p$ products of their entries. Adopting vector-valued RKHS will allow to define a controller with multiple  outputs. The nonparametric nature of RKHSs translates into a rich space of policies that operate in continuous spaces, so that we can avoid discretizing  variables, optimizing over the policies directly instead of working with possibly inaccurate parameterizations.
\end{remark}

Based on its current state, the agent uses its policy to decide its next action to take. While deterministic policies i.e., $a_t=h(s_t)$ are possible, in this work, we focus on stochastic policies. Policies with a random component foster exploration while offering several desirable properties \cite{sutton2000policy}. To this end, we introduce zero-mean multivariate Gaussian noise $n_t$, which renders actions to be  $a_t=h(s_t)+n_t$. Equivalently, we can define the conditional probability of the action given the policy and the state  $\pi_h (a | s ) : \calS \times \calA \to [0,1]$ as
\begin{align}
\pi_h(a|s)=\frac{1}{\det ( 2 \pi \Sigma ) } 
e^{-\frac{1}{2}
\bigl(a - h(s) \bigr)^T \Sigma^{-1} \bigl(a - h(s) \bigr)
 }
\end{align}
where $\Sigma$ is the covariance matrix of the Gaussian distribution. At time $t$, this is equivalent to acting by drawing a sample $a_t \sim \mathcal{N}(h(s_t),\Sigma)$. Ultimately, the objective of the agent is to find a policy $h$ that maximizes the expected discounted returns, given by
\begin{align}	
\label{eq:RLutility}
U(h) \triangleq  \E_{s,a}\left[\sum_{t=0}^\infty \gamma^t r(s_t,a_t)\Big| h\right],
\end{align}  
 where $\gamma \in (0,1)$ is the discount rate. Lower discount rates $\gamma$ favor myopic agents only concerned about the near future, while   $\gamma$ closer to $1$ corresponds to farsighted agents that value future rewards more prominently \cite{sutton2018reinforcement}. The expectation in \eqref{eq:RLutility} is taken with respect to the whole trajectory of states and actions. 
 
\subsection{Multi-task Reinforcement Learning}
\label{subsec:multitaskrl}

Now, we consider the case of multiple tasks. Specifically $N$ tasks, each of them with its expected discounted return $U_i$, for $i=1,\ldots,N$. The goal of our multitask RL methodology is to learn individual $h_i$ for each task and a common policy $g$ with acceptable level of performance for the tasks. It is expected that information from different tasks would enable enhanced learning performance if the reward data associated to these tasks is correlated. A simple approach would be to directly obtain the individual policies $h_i$ that maximize the discounted return of each of the $N$ tasks separately. Namely,
\begin{align}
\tag{PI}
 \{ h_i^\star \}= \underset{h_i\in \mathcal H}{\argmax} ~ U_i (h_i).
   \label{eq:DistributedProblemFormulation}
\end{align}
The problem with this formulation is that it is entirely distributed. Hence, this approach is unable to provide any form of generalization across tasks. In other words, and individual policy $h_i^\star$ obtained by computing \eqref{eq:DistributedProblemFormulation} maximizes the expected discounted reward $U_i$ but it is not trained to achieve any level of performance in $U_j$ for $j \neq i$, rendering no generalization performance across tasks. Furthermore, since, each policy $h_i^\star$ is obtained separately by maximizing a particular expected return $U_i$, the data used for policy $h_i^\star$ does not contribute to the learning of policies $h_j^\star$.

The other possible extreme  would be to force a single policy $g$ to maximize the average expected discounted return $U_i$ over the $N$ tasks, therefore learning a policy that perform well on average for all tasks. That is, 
\begin{align}
g^\star= \underset{g\in \mathcal H}{\argmax}  ~ \frac{1}{N} \sum_{i=1}^N U_i (g).
\label{eq:SinglePolicyProblem}
\tag{PII}
\end{align}
The intuition being this formulation is that while the expected discounted returns of each task $U_i$ might be different from each other, there might exist some correlation among them. However, adopting a common policy $g$  reduces the adaptability of the learning algorithm to each specific tasks. Besides, when policy $g^\star$ is analyzed separately, say for a particular task $U_i$, nothing guarantees individual task performance. 

\subsection{Cross-learning}

The drawbacks of the previous approaches are apparent. They either aim to obtain individual policies, as in \eqref{eq:DistributedProblemFormulation}, or a central policy as in \eqref{eq:SinglePolicyProblem}. In order to leverage experience among tasks and achieve a generalization performance we introduce our cross-learning framework. This formulation aims to learn a family of individual task policies $\bar h_i$ together with a cross-learned policy $\bar g$ by solving the following optimization problem,
\begin{subequations}
\begin{align*}
\tag{PCL}
\label{eq:originalProblemCrosslearning}
   \{\bar h_i\},\bar g=\underset{h_i,g\in\calH}{\argmax}  \quad   & \sum_{i=1}^{N}  U_i \left(h_i\right)  \\
   \text{subject to}
		\quad   & \left\| h_i - g \right\| \leq \epsilon, \quad i=1,\ldots,N.
\end{align*}
\end{subequations}
The key guiding principle of our framework is the search for policies that stay close together in the optimization space via the introduction of a measure of centrality $\epsilon$. This novel formulation allows each policy $h_i$ to train in it specific expected discounted return $U_i$ while forcing centrality through the cross learned policy $g$. Note that since the central policy $\bar g$ is not forced to maximize the sum of expected discounted returns, $g$ will be close to all optimal policies $\bar h_i$ but not necessarily perform on the average of the expected discounted returns as in \eqref{eq:SinglePolicyProblem}.

For values of $\epsilon$ sufficiently large,  the constraint in \eqref{eq:originalProblemCrosslearning} becomes inactive, rendering problems  \eqref{eq:DistributedProblemFormulation} and \eqref{eq:originalProblemCrosslearning} equivalent. Problem \eqref{eq:DistributedProblemFormulation} can be seen as the unconstrained version of problem \eqref{eq:originalProblemCrosslearning}. Likewise, setting $\epsilon=0$ forces a consensus constraint in which all policies $h_i$ are equal one another, yielding problems \eqref{eq:SinglePolicyProblem} and \eqref{eq:originalProblemCrosslearning} equivalent. Otherwise the $h_i$ policies are brought closer together at the cost of reducing each task's utility, in exchange of better generalization performance. 

With the fully distributed formulation \eqref{eq:DistributedProblemFormulation} each agent achieves  a specific-task,  adjusting its policy  $h_i$ separately to optimize its own  expected cumulative reward $U_i$ as in \eqref{eq:RLutility}. In this way, different agents optimize for their own rewards, being agnostic of other agents and their tasks. In contrast, the cross-learning formulation \eqref{eq:originalProblemCrosslearning} couples the rewards through the centrality constraint, bringing policies together and causing the reward information of an agent to spread to the policies of all other agents. The bound $\epsilon$ is a bargaining parameter that tunes how much the goals of others affect each agent policy.

So far we have assumed that the expectations that define each $U_i$ are computable. Otherwise, agents can solve problems \eqref{eq:DistributedProblemFormulation} and \eqref{eq:originalProblemCrosslearning} via  stochastic RL algorithms as given in Section \ref{sec:RL} for the case when the rewards are not available in closed-form and policies are learned from data. Various learning algorithms exist, but they share in common that the policies are improved sequentially by incorporating reward data that the agents collect while they take actions. An agent running an RL algorithm for \eqref{eq:DistributedProblemFormulation} would learn from its own task-specific reward data. Instead, our cross learning formulation yields an RL algorithm in which data acquired  by all agents across tasks percolates through the centrality constraint. Collecting data from multiple sources enables  enhanced learning performance if these sources are correlated, specially at early stages of the learning process when the amount of data collected per task is limited. In addition, the common policy $g$ can also be used as an after-training policy. This would be the case of an algorithm that, after being trained for multiple tasks, needs to quickly adapt to one of these tasks without necessarily knowing which one. In this sense, this is similar to finding a good initialization point of a learning algorithm for a family of tasks $\{U_i\}$.

%% file: 03_reinforcement_learning.tex
\section{Stochastic Policy Gradient}
\label{sec:RL}

Now that we have formulated the cross-learning problem we will attempt to solve it. Two immediate issues need to be addressed before proceeding to obtain the solution to problem \eqref{eq:originalProblemCrosslearning}. First, we will resort to gradient ascent on the objective of the problem, the form of which is given by the policy gradient theorem \cite{sutton2000policy}. Secondly, we need to address the constraints. While we could resort to primal-dual methods\cite{paternain2019constrained, paternain2019safe} to overcome this issue, we prefer to exploit the structure of the set of constraints which form a convex set. Hence, we can compute the projection to this set at each step of the gradient ascent. In this section, we will  focus on the computation of the gradient leaving the projection step for next section. Ignoring the constraints $||h_i-g|| \leq \epsilon$, the update step (pre-projection) for policy $h_i$ at iteration $k$ is given by
\begin{align}
\label{eq:learning_step_deterministic}
	\tilde{h}_{i}^{k}=h_i^k+\eta \nabla_{h_i} U_i (h_i^k) , 
\end{align} 
where $\eta \in (0,1)$ corresponds to the step size. Upon defining the Q-function $Q_i:\calS\times \cal A\to \R$
\begin{align}
\label{eq:q_function}
Q_i(s,a;h_i) \triangleq \E\left[\sum_{t=0}^\infty \gamma^t r_i(s_{i,t},a_{i,t})\Big| h_i, s_{i,0}=s, a_{i,0}=a\right],
\end{align}
it  is shown in \cite{sutton2000policy},\cite{lever2015modelling} that
\begin{align}\label{eq:nabla_U}
     &\left.\nabla_{h_i} U_i(h_i)\right|_{s=s'}= \nonumber\\
&\frac{1}{1-\gamma}\mathbb{E}_{(s,a)\sim \rho_i}\left[ Q_i(s,a;h_i)\kappa(s,s')\Sigma^{-1}\left(a-h_i(s)\right) \Big| h_i\right],
\end{align}
where the expectation in \eqref{eq:nabla_U} is taken with respect to discounted occupation measure,
\begin{align}
\label{eqn_discounted_distribution}
\rho_i(s,a) \triangleq (1-\gamma)\sum_{t=0}^\infty \gamma^t p(s_{i,t}=s,a_{i,t}=a).
\end{align}

The probability density function in \eqref{eqn_discounted_distribution} depends on the transition probabilities of the system as well as the agent's policy. This can be shown by expanding the joint probability, which leads to the following expression.
\begin{align}
\label{eqn_discounted_distribution_expanded}
&\rho_i(s,a)=\frac{1}{1-\gamma} \sum_{t=0}^\infty \gamma^t\pi_{h_i}(a_{i,t}|s_{i,t}) \nonumber\\
&\quad \quad\quad \times \prod_{u=0}^{t-1} p(s_{i,u+1}|s_{i,u},a_{i,u})\pi_{h_i}(a_{i,u}|s_{i,u}) p(s_{i,0}).
\end{align} 
Since the agent is assumed to not have access to the state transition probabilities, the only available resources to the agent are the states, actions and rewards over a trajectory. This means that the expectation in $U_i (h_i)$ as well as its gradient $\nabla_{h_i} U_i(h_i)$ are not computable. However, we can construct an unbiased estimate $\hat{\nabla}_{h_i} U_i$ using the available states $s_{i,t}$, actions $a_{i,t}$ and rewards $r_{i,t}$. Two challenges arise in the search for such an estimate. First,  we need to estimate the Q-function $Q_i(s,a)$, and secondly, we need to sample from the distribution $\rho_i(s,a)$. 

An unbiased estimate $\hat Q_i(s,a, T)$ can be found by sampling a geometric variable $T$ with parameter $\gamma$, such that $P(T=0)=1-\gamma$, and adding $T$ consecutive rewards \cite{bertsekas1996NDP},\cite{paternain2018stochastic}  
\begin{align}
\label{eq:Qhat}
\hat Q_i(s,a, T)=\sum_{\tau=0}^T r_{i,\tau}.
\end{align}
These rewards are collected by an agent following a trajectory of $T$ steps, starting from state $s$ and action $a$, and using policy $h(s)$ for selecting actions henceforth. 

\begin{algorithm}[t]
	\caption{Unbiased estimate  for the gradient of $U_i$.}
	\label{alg:unbiased}
	\begin{algorithmic}[1]
		\State  Draw $t,T\sim\textrm{Geom}(\gamma)$
		\State  Run the system up to time $t$
		\State  Save $s_{i,t}$
		\State  Compute $a_{i,t}=h(s_{i,t})+n_{i,t}$
		\State  Compute $\hat Q_{i,t}=\hat Q_i(s_{i,t},a_{i,t}, T)$
		\State  Compute $\hat w_{i,t}=\Sigma^{-1}(a_{i,t}-h_i(s_{it}))(\hat Q_{i,t})/(1-\gamma)$
		\State  Obtain $\hat{\nabla}_{h_i} U_i= \Kappa(s_{i,t},\cdot)\hat w_{i,t}$
	\end{algorithmic}
\end{algorithm}

Next we use \eqref{eq:Qhat} to find an unbiased estimate of the gradient sampling from \eqref{eq:nabla_U}, i.e., 
\begin{align}\label{eq:hat_nabla_U}
&\left.\hat \nabla_{h_i} U_i(h_i)\right|_{s=s'}=\frac{1}{1-\gamma}\hat  Q_i(s,a,T)\kappa(s,s')\Sigma^{-1}\left(a-h_i(s)\right).
\end{align}
For \eqref{eq:hat_nabla_U} to be unbiased, the pair $(s,a)$ must be sampled from the distribution in \eqref{eqn_discounted_distribution}. The procedure to find such a pair  is shown in Steps $1-4$ of Algorithm \ref{alg:unbiased}. Starting from  $(s_{i,0},a_{i,0})$ and selecting the actions according to $h_i(s)$, the agent runs $t$ steps of a trajectory  with $t$ drawn  from the same  geometric distribution again. The final pair  $(s,a)=(s_{i,t},a_{i,t})$ follows the desired distribution $\rho(s,a)$  \cite{paternain2018stochastic}.

The remaining steps of Algorithm \ref{alg:unbiased} are simply an ordered list of instructions to obtain   $\hat \nabla_{h_i} U_i$ in \eqref{eq:hat_nabla_U}.
Intuitively, selecting a geometric random variable yields a future time instant such that $s_{i,t}$ is representative of the whole trajectory.  The weight $\hat w_{i,t}$ represents a stochastic sample of the gradient around $s_{i,t}$.  The estimate includes a kernel function $\kappa(s_{i,t},s')$ which interpolates the gradient to nearby states $s'$.
This results in the unbiased estimate of $\hat{\nabla}_{h_i} U_i$ as stated in the following proposition, which is brought  from \cite{paternain2018stochastic}.
\begin{proposition}\label{prop:stochastic_gradient}
The stochastic gradient $\hat{\nabla}_{h_i} U_i$ is unbiased; i.e.,  $\E\bigl[\hat{\nabla}_{h_i} U_i\bigr]=\nabla_{h_i} U_i$.
\end{proposition}
\begin{proof}
The proof can be found in \cite{paternain2018stochastic} for a generalized version of Algorithm \ref{alg:unbiased} with two-sided weights.
\end{proof}
Using this result we run a stochastic gradient ascent iteration, given by the following expression
\begin{align}\label{eq:learning_step_deterministic}
\bar{h}_{i}^{k}=h_i^k+\eta^k \hat{\nabla}_{h_i} U_i (h_i^k).
\end{align} 
For each step, a new element $\eta^k \hat{ \nabla}_{h_i} U_i$ is added to policy $h_i$ in order to obtain policy $\bar{h}^k_i$. More specifically, the dictionary and weights  of policy $\bar{h}_i$ have new  elements, $s_{i,t}$ and $w_{i,t}=\eta \hat w_{i,t}$, respectively. State $s_{i,t}$ stands as a  knot that specifies a new kernel which is combined into the policy. This procedure presents a major drawback due to memory explosion, with infinite kernels expanding the policy.  This problem is addressed by the pruning procedure discussed in the Appendix \ref{app:cdkomp}. The procedure is based on \cite{vincent2002kernel} and discards kernels that are not informative.

While  \eqref{eq:hat_nabla_U}  is provably unbiased    and  \eqref{eq:learning_step_deterministic}  has theoretical convergence guarantees \cite{paternain2018stochastic}, the variance of  \eqref{eq:hat_nabla_U}  may be high reflecting on a slow  speed of convergence of \eqref{eq:learning_step_deterministic}.  For these reasons, the literature on RL includes several practical improvements. Variance can be reduced by including batch versions of the  gradient method, in which several stochastic gradients are averaged before performing the update in \eqref{eq:learning_step_deterministic}. One particular case of a batch gradient iteration in   \cite{paternain2018stochastic}, averages two gradients sharing the same state $s_{i,t}$ but with actions $a_{it}$ and $a'_{it}$  in opposite directions. This results in a  gradient that favors the action with higher $\hat Q(s_{i,t},a)$. Similarly, using estimates of the advantage function $A(s,a)=Q(s,a)-V(s)$ with respect to the value function $V(s)=E_a[Q(s,a)]$  the gradient favors actions such that $Q(s,a)$ outperforms its average \cite{silver2014deterministic}.     
Besides these variants for estimating $\nabla_{h_i} U_i$, there are several flavors of the stochastic algorithm \eqref{eq:learning_step_deterministic} that combine subsequent gradients in the update for the purpose of introducing inertia. A prominent example of such an algorithm is the  ADAM optimizer  \cite{kingma2014adam}.

With  any of these RL variants being selected, we proceed to combine the stochastic gradients across tasks for a cross-learning algorithm that solve the problem \eqref{eq:originalProblemCrosslearning}.

%% file: 04_geometric_interpretation.tex
\section{Projected Policy Gradient}
\label{sec:GeoInter}

In the previous section, we studied the required steps to compute the gradient update  for  the cost of the cross-learning problem \eqref{eq:originalProblemCrosslearning}. Subsequently, the policy that results from the gradient update needs to be projected to the  constraint set in \eqref{eq:originalProblemCrosslearning}. To this aim, we dedicate the section to this projection.

Let  $\{\bar{h}_i,\bar{g}\}$ denote the set of policies that result from the gradient step  in \eqref{eq:learning_step_deterministic}. Notice that since the  cross-learned policy $g$ is not a variable in the cost, it  remains unchanged after the gradient step. Let $\calP_{\calC}\left[\bar h_1,\ldots,\bar h_N,\bar g\right]$ denote the projection of $\bar h_1,\ldots,\bar h_N,\bar g$ into the set of constraints $\calC=\{h_i,g\in\calH:\ \|h_i-g\|\leq \epsilon\}$, which is defined by
\begin{subequations}
\begin{align}
  \calP_{\calC}\left[\{\bar h_i\},\bar g\right] =&\underset{\{h_i\},g}{\argmin}  \hspace{-.1cm} \quad    \sum_{i=1}^{N} \left\| h_i - \bar{h}_i \right\|^2 +\left\| g - \bar{g} \right\|^2   \\
   &\text{subject to} 		\quad    \left\| h_i - g \right\|^2 \leq \epsilon^2, \quad i=1,\ldots,N.
\end{align}
\label{eq:projectionProblemC}
\end{subequations}
Given $\calP_{\calC}$, the projected stochastic gradient ascent update for problem \eqref{eq:originalProblemCrosslearning} with step size $\eta^k$ takes the form
\begin{align}
(\{h_i^{k+1}\},g^{k+1})=\calP_{\calC}\left[\left\{ h^k_i+\eta^k \hat  \nabla_{h_i} U_i(h_i^k)\right\}, g^k\right],
\label{eq:PGD}
\end{align}
which takes into account that the gradient of  $U_i$ with respect to $g$ is null. The following lemma, which will be called upon throughout the section, establishes that the solution $h_i$  lays inline between  $\bar h_i$ and $g$, and establishes a useful equation.  
\begin{lemma}\label{lemma:inline}
If $(\bar g,\bar h_1,\ldots,\bar h_N)$ and $( g, h_1,\ldots, h_N)$ are the input and output of the projection operator  $\calP_{\calC}$ then 
\begin{align}\label{eq:convex_comb}
h_i&= \zeta_i \bar h_i +(1-\zeta_i)  g
\end{align}
with $\zeta_i \in [0,1]\subset \R,$
and 
\begin{align}
g-\bar g&=  \sum_{i=1}^N (\bar h_i -h_i).\label{eq:deltag}
\end{align}
\end{lemma}
\begin{proof}
Writing the Lagrangian for \eqref{eq:projectionProblemC} with scalar multipliers $\mu_i\geq 0,\ i=1,\ldots,N$ and setting its derivatives with respect to $g$ and $h_i$ to zero, it yields
\begin{align}
h_i-\bar h_i+\mu_i (h_i-g)&=0\label{eq:nablaLhi}\\
g-\bar g+\sum_{i=1}^N \mu_i(g-h_i)&=0 \label{eq:nablaLg}. 
\end{align}
Equation \eqref{eq:nablaLhi} can be rearranged into $h_i=\frac{1}{1+\mu_i} \bar h_i +\frac{\mu_i}{1+\mu_i} g$, which coincides with \eqref{eq:convex_comb} for $\zeta_i=1/(1+\mu_i)$. Also, adding \eqref{eq:nablaLhi} over $i$ and substituting the sum in \eqref{eq:nablaLg} results in \eqref{eq:deltag}.
 \end{proof}     

\begin{figure}[t]
	\centering
    \includegraphics[scale=1]{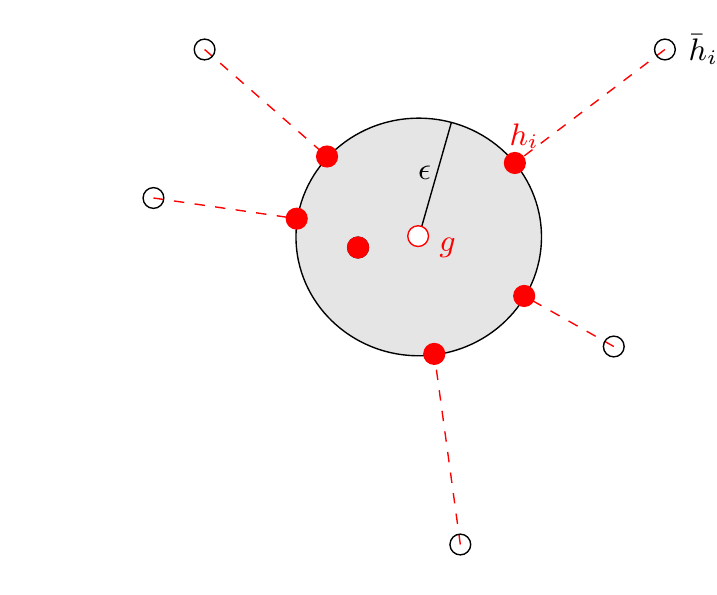}        	
	\caption{Illustration of the effect of the projection $\calP_{\calC}$.}
\label{fig:projC}
\end{figure}

Fig. \ref{fig:projC} shows the effect of  the projection $\calP_{\calC}$ in a simplified case. For the illustration  purpose of Fig. \ref{fig:projC}, replace the RKHS $\calH$ by the standard vector space  $\R^2$. In this case we solve \eqref{eq:projectionProblemC} and obtain that  vectors $h_i$ lay at the intersection of the line connecting  $\bar h_i$ with  $g$ (cf., \eqref{eq:convex_comb}), and the circumference of center $g$ and radius $\epsilon$. Policies $h_i$ are brought together close to the common center $g$, half-way between the fully separable version \eqref{eq:DistributedProblemFormulation} in which $h_i=\bar h_i$, and the strict solution \eqref{eq:SinglePolicyProblem} in which all the $h_i$ would be forced to be equal.

Returning to the  RKHS case in which  $g, h_i\in \calH$, the following proposition demonstrates  how to project $\bar h_i$ and $\bar g$ into $\calC$ by solving a standard  Quadratically Constrained Quadratic Program (QCQP). We focus on the case $p=1$ since a generalization is simple but the notation more involved. The following  assumption, which   will be satisfied by the RL algorithm, states that the $\bar h_i$ are finite dimensional, they share $L$ kernels $s_1,\ldots,s_L$ with $\bar g$, and differ  in a unique kernel $s_{L+i}$.

\begin{assumption}
\label{assumption:kernels}
Functions $\bar g$ and $\bar h_i$ are given by finite sums of weighted kernels with the following structure $\bar h_i=\sum_{m=1}^{L+N} \bar a_{mi} \kappa(s_m,s)$,  where  $\bar a_{mi}=0$ for $m>L,\ m\neq L+i$ and  $\bar g=\sum_{m=1}^{L+N} \bar c_{m} \kappa(s_m,s)$ where $\bar c_m=0$ for $m>L$.
\end{assumption}

 For notation brevity, collect the coefficients  $\bar a_{mi}$ and $\bar c_m$ in  a matrix $\bar A\in\R^{(L+N)\times N}$ and a vector $\bar c\in\R^{(L+N)}$, respectively. Also define  $1_N\in \R^N$ as the vector of all ones, and  $K\in\R^{(L+N)\times(L+N)}$ as the kernel Gram matrix with elements $\kappa(s_m,s_{m'})$. Under this definitions the projection can be solved as a QCQP, as detailed in the following proposition.
\begin{proposition}\label{prop:solution_projC}
If  $\bar h_1,\ldots,\bar h_N, \bar g\in \calH$ satisfy Assumption \ref{assumption:kernels}, then  the solution to \eqref{eq:projectionProblemC}   is given by 
	\begin{align}
 h_i&=\sum_{m=1}^{L+N}  a_{mi} \kappa(s_m,s),\quad 
 g=\sum_{m=1}^{L+N}  c_{m} \kappa(s_m,s),\label{eq:handg}
\end{align}
where the coefficients $a_{mi}$, $ c_{m}$, are the entries of the matrix $A\in\R^{(L+N)\times N}$ and vector $c\in\R^{(L+N)}$  that solve the following problem
\begin{subequations}
\begin{align}
   \underset{
   \begin{subarray}{c}
  		A \in \R^{(L+N)\times N} \\
  		c\in\R^{L+N} \vspace{-0.5cm}
  	\end{subarray}    \hspace{-.4cm} 
   }
{\textup{minimize}}  \quad   & \tr\left({(\bar A-A)^T K (\bar A-A)}\right)+(\bar c-c)^T K (\bar c-c)  \\ 
   \textup{subject to} 
		\quad   & \hspace{-.1cm} \diag\left((A-c1_N^T)^TK(A-c1_N^T)\right)\leq \epsilon^2 1_N . \hspace{-.15cm} 
\end{align}
\label{eq:QCQP}
\end{subequations}

\end{proposition}

\begin{proof}It follows from Lemma \ref{lemma:inline} that $h_i$ and $g$ are linear combinations of  $\bar h_i$ and $\bar g$, and thus they   admit the finite dimensional expansions in \eqref{eq:handg}. 
The QCQP in \eqref{eq:QCQP} follows from expanding the norms in \eqref{eq:projectionProblemC} according to \eqref{eq:handg}.
\end{proof}

\subsection{Relaxed Projection}
\label{subsec:relaxed}

As per Proposition \ref{prop:solution_projC}, in order to solve the cross-learning problem \eqref{eq:originalProblemCrosslearning}, we need to solve a QCQP at each gradient step, which may become computationally expensive for policies represented with a large amount of kernels. For a simplified closed-form solution to the projection \eqref{eq:projectionProblemC}, we consider the following relaxation of the cross-learning problem \eqref{eq:originalProblemCrosslearning} which we denote as \emph{relaxed cross-learning},
	\begin{subequations}
\begin{align*}
\tag{RCL}
\label{eq:relaxedProblemCrosslearning}
   \underset{h_i,g\in\calH}{\text{maximize}}  \quad   & \sum_{i=1}^{N}  U_i \left(h_i\right)  \\
   \text{subject to}
		\quad   & \sum_{i=1}^N \left\| h_i - g \right\|^2 \leq N \epsilon^2.
\end{align*}
\end{subequations}
The main difference between the original cross-learning problem \eqref{eq:originalProblemCrosslearning} and its relaxed counterpart \eqref{eq:relaxedProblemCrosslearning} is the number of constraints. In the relaxed cross-learning problem \eqref{eq:relaxedProblemCrosslearning} there is only one constraint, which is specified as the average over the $N$ constraints of the original problem. This results in the reduction to a single Lagrange multipliers, which renders a closed-form solution that can be computed exactly without the need of solving an optimization problem numerically. The associated projection of the relaxed cross-learning problem is,
\begin{subequations}
\begin{align}
  \calP_{\calR}\left[\{\bar h_i\},\bar g\right] =&\underset{h_i,g}{\argmin}  \quad    \sum_{i=1}^{N} \left\| h_i - \bar{h}_i \right\|^2  \\
  & \text{subject to}
		\quad       \sum_{i=1}^N \left\| h_i - g \right\|^2 \leq N \epsilon^2.
\end{align}
\label{eq:projectionProblemR}
\end{subequations}

Notice that we have omitted the term $\left\| g - \bar g \right\|^2$ in the cost of \eqref{eq:projectionProblemR}. This simplification plays an important role in obtaining a simpler closed form solution. Notice also that the set of constraints $\calR=\{h_i,g,\in\calH:\    \sum_{i=1}^N \left\| h_i - g \right\|^2 \leq N \epsilon^2\}$ contains the previous set $\calC$, namely, $\calR$ is a relaxation of $\calC$. 

\begin{figure}[t]
	\centering
    \includegraphics[scale=1]{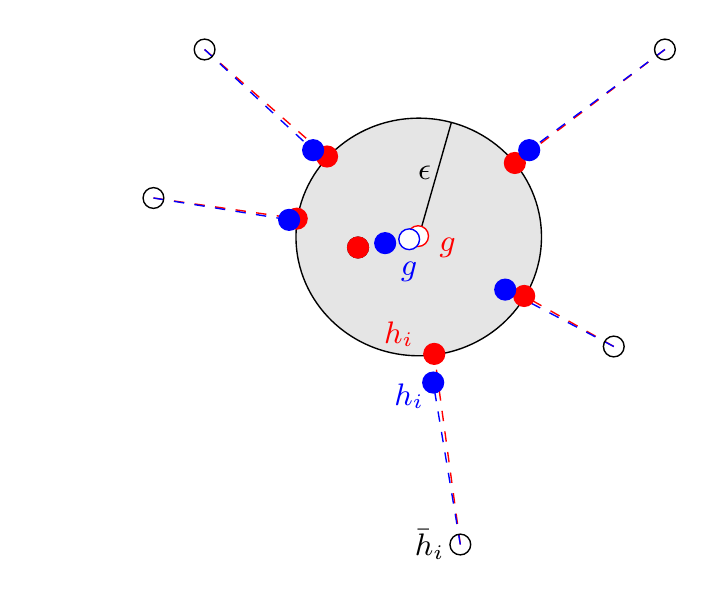}        		
	\caption{Solution of the relaxed projection $\calP_{\calR}$, shown in blue; together with the solution of the original projection $\calP_{\calC}$, shown in red.}
	\label{img:projR}
\end{figure}
 
In Fig. \ref{img:projR} we can see the solution of relaxed projection $\calP_\calR$ given by equations \eqref{eq:projectionProblemR} superimposed to the solution of the original projection $\calP_\calC$ given by equations \eqref{eq:projectionProblemC}. Notice that the centers $g$ resulting from both projections are shifted one from another, and that the points $h_i$ in the relaxed version are not at distance $\epsilon$ from the center policy, as it results when using the original projection $\calP_\calC$.  

Now, after qualitatively comparing the original projection $\calP_\calC$ to its relaxed version $\calP_\calR$ , we will delve into the closed-form solution of the latter. As given in \eqref{eq:projectionProblemR}, $\calP_\calR$ admits the following  closed-form solution. 

\begin{proposition}\label{prop:relaxed}
Given  $\bar h_1,\ldots,\bar h_N \in \calH$,  the solution of \eqref{eq:projectionProblemR} takes the form
\begin{align}
 g&=\frac{1}{N}\sum_{i=1}^{N}  \bar h_i,\quad   h_i=(1-\psi) g +\psi \bar h_i,
 \label{handgR}
\end{align}
where $ \psi=\min\left\{1,N\epsilon\left(\sum_{i,j=1}^N\|\bar h_i-\bar h_j\|^2 \right)^{-\frac{1}{2}}\right\}$.

\end{proposition}
\begin{proof}
Equations \eqref{handgR} follow from setting to zero the derivatives of the Lagrangian for \eqref{eq:projectionProblemR} with multiplier $\mu=1/(1-\psi)$. Then  $\psi=1$ holds when the constraint of \eqref{eq:projectionProblemR} is inactive, and  $\psi=N\epsilon\left(\sum_{i,j=1}^N\|\bar h_i-\bar h_j\|^2 \right)^{-\frac{1}{2}}$    is obtained   by imposing the constraint to be active and substituting $\|\bar h_i -g\|^2 = \frac{1}{N} \left\| \sum_{j=1}^N(\bar h_i-\bar h_j)\right\|^2$.
\end{proof}

As in the case of the original projection $\calP_\calC$ (cf. Lemma \ref{lemma:inline}), the solutions $h_i$ of the relaxed cross-learning projection $\calP_\calR$ lie on the inline between $\bar h_i$ and $g$. 

\begin{remark}
From Propositions  \ref{prop:solution_projC} and \ref{prop:relaxed} it follows that the output of the projections, both the original one and its relaxed version,  renders policies that share the same kernels. This observation is  instrumental for algorithmic reasons, as it allows us to keep a common dictionary of kernels for all policies. Nevertheless, the size of such dictionary grows linear with the number of iterations. This problem must be handled adequately as otherwise it would lead to memory explosion. The details of this problem and the proposed solution are discussion on the following Section \ref{subsec:sparse_projection}.
\end{remark}

We finish this section by comparing the solutions of the original and relaxed projections, $\calP_\calC$ given by \eqref{eq:projectionProblemC} and $\calP_\calR$ given by \eqref{eq:projectionProblemR}, to corroborate the observations in Fig. \ref{img:projR}. Specifically, we prove that if the input center policy $\bar g$ is close to the mean of the input policies, then the two projected centers are close one each other and the difference between the individual policies is bounded.

\begin{proposition}\label{prop:centers}
	Let $(g,h_1,\ldots,h_N)$ and $(g',h'_1,\ldots,h'_N)$ be the solutions to problems \eqref{eq:projectionProblemC} and \eqref{eq:projectionProblemR}, respectively. Then, if $\|\bar g-g'\|\leq \epsilon$, it follows that $\|g-g'\|\leq \epsilon$. Moreover, the individual policies satisfy 
	\begin{align}
	\|h_i-h'_i\|&\leq\epsilon \left( 2 + \sqrt{ \frac{N y_i}{\sum_i y_i}}\right)\label{eq:prop_policies_deltah} 
	\end{align}
	where $y_i^2=\|\bar h_i -g'\|^2 = \frac{1}{N} \left\| \sum_{j=1}^N(\bar h_i-\bar h_j)\right\|^2$.   
\end{proposition}
\begin{proof}
According to Lemma \ref{lemma:inline} $g$, $h_i$ and $\bar h_i$ are aligned, hence we can write  
\begin{align}
g&=h_i+\delta_i v_i\label{eq:prop_centers_g} 
\end{align}
where $v_i=(g-\bar h_i)/\|g-\bar h_i\|$ has norm one and  $\delta_i\leq \epsilon$, with equality if the constraint is active.
Define $u=\bar g-g'$ such that  $\|u\|\leq \epsilon$ by assumption.  
Combining \eqref{eq:deltag} with \eqref{eq:prop_centers_g}, together with $g'= \frac{1}{N}\sum_{i=1}\bar h_i$ and   $\bar g-g'=u$, it follows
\begin{align}
g-g'&=g-\bar g + \bar g -g'=\sum \bar h_i -\sum h_i +u\\
&= N g' -   \sum_i  (g -\delta_i v_i) +u \label{eq:prop_centers_deltag}. 
\end{align}
After rearranging terms, the triangle inequality yields
 $\|g-g'\|\leq \frac{1}{N+1} \left(\|u\| + \sum_i \delta_i \|v_i\| \right) \leq \epsilon$,  as desired. 

Using the bound $\|g-g'\|\leq \epsilon$  just found, we proceed to bound the difference between the individual policies.   
Adding and subtracting $g$ and $g'$ and then applying the triangle inequality we obtain
\begin{align}
\|h_i-h'_i\|&\leq \|h_i-g\| + \|g-g'\| +\|g'-h'_i\|.\label{eq:prop_policies_hi} 
\end{align}
The first term  satisfies $ \|h_i-g\|\leq \epsilon$ because $h_i$ and $g$ are feasible. The second term is also lower than $\epsilon$, since $\|g-g'\|\leq \epsilon$, as previously found. To conclude the proof, observe that according to \eqref{handgR}  $\|g'-h'_i\|^2=\psi^2 {y_i}^2\leq N\epsilon^2 \frac{y_i^2}{\sum_i y_i^2}$. 
\end{proof}

%% file: 05_algorithm.tex
\section{Cross-learning Algorithm for Reinforcement Learning}
\label{sec:CLRLA}

By the combination of both the gradient step presented in Section \ref{sec:RL} and the projection discussed in Section \ref{sec:GeoInter}, we obtain the cross-learning algorithm. The resulting strategy proposed in this paper involves combining the information collected by all agents and projecting them into the set of constraints $\calC$. Furthermore, we substitute the gradient  ${\nabla}_{h_i} U_i$  in \eqref{eq:PGD} for a stochastic estimate $\hat{\nabla}_{h_i} U_i$ obtained by following Algorithm \ref{alg:unbiased}, which results in a stochastic version of the projected gradient ascent method; i.e., 
\begin{align}
(\{h_i^{k+1}\},g^{k+1})&=\calP_{\calC}\left[\left\{ h^k_i+\eta^k \hat{\nabla}_{h_i} U_i(h_i^k) \right\}, g^k\right].\label{eq:PSGA}
\end{align}

The overall resulting method is summarized in Algorithm \ref{alg:CLA}. The cross-learning procedure first collects all the stochastic gradients of the expected reward for each task following the task specific policy $\hat \nabla_{h_i}U_i$.  Then, the policies are brought together to the distance constraints by the projection $\calP_\calC$. The projection $\calP_{\calC}[\cdot]$ into set $\calC$ can be replaced by its relaxed counterpart  $\calP_{\calR}[\cdot]$ if computation resources are limited. Afterwards, Algorithm \ref{alg:CLA} has a step to reduce dictionary complexity, named Common Dictionary Kernel Orthogonal Matching Pursuit (CDKOMP), which we discuss next in Section \ref{subsec:sparse_projection}. The algorithm has a stopping condition given by $\langle \hat \nabla_{h_i}U_i (h_i^{k+1}), (h-h_i^{k+1}) \rangle\leq \alpha$.

\begin{algorithm}[t]
    \caption{Cross-learning algorithm}
    \label{alg:CLA}
    \begin{algorithmic}[1]
  	\State   Initialize $g(\cdot)=0,$ $h^0_i(\cdot)=0,\ i=1,\ldots,N$
	\Repeat \ for   $k=0,1,\ldots$
	\For {$i=1,\ldots,N$}
	\State  Obtain $\hat \nabla_{h_i}U_i$ via Algorithm \ref{alg:unbiased}.  
	\EndFor
	\State Project  $(\{h_i^{k+1}\},g^{k+1})=\calP_{\calC}\left[\{ h^k_i+\eta^k \hat \nabla_{h_i}U_i\}, g^k\right]$ 
	\State Reduce dictionary elements $\text{CDKOMP}(\{h_i^{k+1}\},g^{k+1})$	
   \Until  $\langle \hat \nabla_{h_i}U_i (h_i^{k+1}), (h-h_i^{k+1}) \rangle\leq \alpha,$ $\forall i, \ \forall h \in \calC$
  	\end{algorithmic}
  	
\end{algorithm}

\subsection{Reducing the Dictionary Complexity}
\label{subsec:sparse_projection}

The cross-learning algorithm suffers from a practical drawback that need to be addressed. New kernels are added per policy during each step of the gradient ascent algorithm, leading to memory explosion.  This issue arises due to the stochastic gradient produced by Algorithm \ref{alg:unbiased}, which results in the addition of a new kernel $\kappa(s_{i,t},\cdot)$ to the dictionary of policy $\pi_i$ per gradient step. Therefore, there are $N$ new kernels per itaration, one per policy. Furthermore, by performing linear combinations of all policies, the projection \eqref{eq:projectionProblemC} (or its relaxed counterpart \eqref{eq:projectionProblemR}) add all these  new knots to the dictionary of each single policy.

To overcome this limitation, we reduce the dictionary elements via Common Dictionary Kernel Orthogonal Matching Pursuit (CDKOMP). This procedure performs an additional projection to a lower dimension at each gradient step which both keeps the memory bounded and merges the policies dictionaries. The less informative kernels  in the dictionaries are discarded, keeping only  those that expand the current policies $h_i^k$ and $g^k$ up to a prescribed level of precision $\beta$. The algorithm, which is based on Kernel Orthogonal Matching Pursuit \cite{vincent2002kernel,koppel2019parsimonious}, induces bias that is bounded at all times. Specifically, the CDKOMP introduces bias terms $b_i^k$ to policies $h_i^{k+1}$ and $b^k_{N+1}$ to policy $g^{k+1}$ all bounded at all steps, i.e., 
\begin{align}
(\{h_i^{k+1}\},g^{k+1})&=\calP_{\calC}\left[\left\{ h^k_i+\eta^k \hat{\nabla}_{h_i} U_i(h_i^k) \right\}, g^k\right] 
\nonumber\\ &+ (\{b^k_i\},b_{N+1}^k),
\end{align}
with $||
b^k_i||<\beta, k=1,\dots,K, i=1,\dots,N+1$. An in-depth explanation of the procedure can be found in Appendix \ref{app:cdkomp}.

%% file: 06_convergence_analysis.tex
\section{Convergence of the Cross-learning Algorithm}
\label{sec:convergence}

In this section we provide  convergence guarantees of the cross-learning Algorithm \ref{alg:CLA} to a \emph{near-optimal} point with high probability. 

\subsection{Properties and Assumptions}
\label{subsec:asumptions}

As a base for convergence analysis, we introduce a set of assumptions that need to be satisfied for all tasks.

\begin{assumption}\label{As:Max_Grad_Deterministic}
The norm of the gradient of the expected discounted returns $\nabla_{h_i} U_i$ is upper bounded by a constant $\mu_i$, i.e., for any policy $f$
\begin{align}
||\nabla_{h_i} U_i (f)|| \leq \mu_i.
\end{align}	
\end{assumption}

\begin{assumption}\label{As:Grad_Lipschitz}
The gradients of the expected discounted return $\nabla _{h_i} U_i$ are Lipschitz continuous with constant $L_i$, i.e., for any two  policies $f_1$,$f_2$, 
\begin{align}
||\nabla_{h_i} U_i (f_1)-\nabla_{h_i} U_i (f_2)||\leq L_i ||f_1-f_2||. 
\end{align}
\end{assumption}

\begin{assumption}\label{As:Variance_Grad}
The expected norm of the square of the stochastic gradient of the expected discounted return $U_i$ is upper bounded, i.e., for any policy $f$
\begin{align}
\EX [||\hat{\nabla}_{h_i} U_i (f)||^2]\leq B^2_{U_i}.
\end{align}
\end{assumption}

\begin{assumption} \label{As:DifferenceGradients}
	Given a batch size number $b_{U_i}$ used to obtain the stochastic gradient of the expected discounted return $\nabla_{h_i} U_i$, its second moment is upper bounded, i.e., for any policy $f$
	\begin{align}
	\E[||\hat{\nabla}_{h_i}U_i(f) - \nabla_{h_i} U_i(f)||^2] \leq \frac {\sigma_{U_i}^2}{b_{U_i}}.
	\end{align}
\end{assumption}

Notice that all of these statements are standard assumptions when dealing with constrained non-convex stochastic optimization problems. Furthermore, they can be shown to be satisfied by  assuming that the reward functions $r_i(s, a)$ are upper bounded \cite{paternain2018stochastic} . 

Under these assumptions the proposed method takes the form of a nonconvex projected stochastic gradient ascent algorithm. Therefore, the following convergence analysis for cross-learning makes use of the usual approach to the analysis of these kinds of algorithms. Specifically, we follow an approach similar to that of \cite{mokhtari2018escaping}, with the resulting statements tailored to our framework, including the bias introduced by the CDKOMP Algorithm \ref{alg:MDKOMP}, the multi-task formulation and the lack of gradient update for the cross-learning policy $g$.

\subsection{Convergence Analysis}  

For convenience, we begin by defining the cross-learning function $F : \calH^N \to \real$, which corresponds to the objective function of the cross-learning problem \eqref{eq:originalProblemCrosslearning}. This function receives a vector of policies $\textbf{h}=[h_1,\dots,h_N]$ and returns the sum of the expected discounted returns of each policy $U_i(h_i)$. Namely,
\begin{align}\label{eq:def_F}
F(\textbf{h})=\sum_{i=1}^N U_i(h_i).
\end{align}
Further, we define the gradient of the the cross-learning function $F$ with respect to the policy vector $\textbf{h}$ is given by a vector of the gradients of each expected discounted return $U_i$ with respect to each policy $h_i$, i.e.,
\begin{align}
\nabla_{\textbf{h}} F=[\nabla_{h_1}U_1,\dots,\nabla_{h_N}U_N].
\end{align}
Now, let us recall the standard criterion for local maxima (also denoted as a first-order stationary point) over a convex set of constraints $\calC$.
\begin{proposition}\label{prop:LocalMaximum}
If $h_i^\star$ is a local maximum of the function $U_i$ over the convex set of constraints $\calC$, then
\begin{align}\label{LocalMaximum}
\langle \nabla_{h_i} U_i (h_i^\star), (h-h_i^\star)\rangle \leq 0, \quad \forall h \in \calC.
\end{align}	
\end{proposition}
\begin{proof}
See e.g., \cite[Proposition 2.1.2]{bertsekas1999nonlinear}.
\end{proof}
Notice that, under this condition, if the local maximum of the function $U_i$ is attained by a policy $h_i^\star$ inside of the convex set of constraints $\calC$, then the gradient must be zero, i.e. $\nabla_{h_i} U_i (h_i^\star)=0$. In the case that  $h_i^\star$ is not in the interior of the set of constraints $\calC$, then $\langle \nabla_{h_i}  U (h_i^\star), (h_i-h_i^\star)\rangle < 0$. 

Ideally, the objective would be for Algorithm \ref{alg:CLA} to achieve a local maximum, for the sum of expected discounted returns of all policies, namely, the cross-learning function defined in equation \eqref{eq:def_F}. However, there are practical issues that need to be taken into account. First, it may not be possible to converge to $h_i^\star$ in a finite number of steps, rendering problem \eqref{eq:originalProblemCrosslearning} unsolved. On the other hand, if an accuracy level $\alpha$ is admissible, then a neighborhood of the solution $h_i^\star$ may be reached in a moderate number of iterations, and the cross-learning Algorithm \ref{alg:CLA} can be stopped. Thus, we introduce the following definition.  
\begin{definition}\label{def:alphafosp}
Given $\alpha>0$,  $h_i^{K}$ is a $\alpha$-First-Order Stationary Point ($\alpha$-FOSP) of the function $U_i$ over the convex set $\calC$ if
\begin{align}\label{def:eqalphafosp1}
\langle \nabla_{h_i} U_i (h_i^{K}), (h-h_i^{K}) \rangle \leq \alpha, \quad \forall h \in \calC. 
\end{align}
\end{definition}

Notice that Definition \ref{def:alphafosp} is an $\alpha$-relaxation of the standard criterium of a local maximum, rendering equal solutions $h_i^{K}=h_i^\star$ if $\alpha=0$ (cf. Proposition \eqref{prop:LocalMaximum}). Still, the $\alpha$-FOSP condition is not readily met by the stopping condition of the cross-learning algorithm. Since we are unable to compute the value of the expected discounted returns $U_i$, nor obtain its gradient $\nabla_{h_i} U_i$,  Algorithm \ref{alg:CLA} utilizes an unbiased estimate of the gradient of $U_i$, namely $\hat{\nabla}_{h_i} U_i$, as the stopping condition. Hence, given $\alpha >0$, the cross-learning algorithm will stop at iteration $k=K$ if for all tasks, the following condition is met,  
\begin{align} \label{eq:stochasticalphafosp}
\langle \hat{\nabla}_{h_i} U_i (h_i^K), (h-h_i^K)\rangle \leq \alpha, \quad  \forall h \in \calC. 
\end{align}
Notice that equation \eqref{eq:stochasticalphafosp} is the stochastic version of equation \eqref{def:eqalphafosp1}, meaning that for condition \eqref{eq:stochasticalphafosp}, the stochastic gradient of $\nabla_{h_i} U_i $, given by $\hat{\nabla}_{h_i} U_i$ is used instead. Another way of expressing the stopping condition of is by its maximum inner product. Hence, the stopping condition of the cross-learning algorithm can be equivalently given by,
\begin{align}
	\max_{\underset{h \in \calC}{i \in \{1,\dots,N\}}}    \langle \hat{\nabla}_{h_i} U_i (h_i^K), (h-h_i^K)\rangle \leq \alpha.
\end{align} 
Notice that imposing the condition \eqref{eq:stochasticalphafosp} for all tasks $i=1,\ldots,N$ is not the same as using the cross-learning function $\langle \hat{\nabla}_{\textbf{h}} F(\textbf{h}^K),(\textbf{h}-\textbf{h}^K)\rangle \leq \alpha$, since we need to ensure that the condition is satisfied for each task. 

Now, the challenge arises when using the stochastic condition \eqref{eq:stochasticalphafosp} to stop the cross-learning algorithm, because the $\alpha$-FOSP condition \eqref{def:eqalphafosp1} is defined for  the  deterministic gradient. In the following, we show that if condition \eqref{eq:stochasticalphafosp} is met for a given $\alpha>0$ and a set of $i =1,\dots,N$ tasks by a sequence of $K$ iterates generated by Algorithm \ref{alg:CLA}, then with high probability the resulting policy $\textbf{h}^K$ is an $\alpha$-FOSP of problem \eqref{eq:originalProblemCrosslearning}. 

We start by showing that before the stopping condition is met, the expected difference between the cross-learning function $F$ at two subsequent iterations $\textbf{h}^k$ and $\textbf{h}^{k+1}$, with $k<K$, is lower bounded.

\begin{lemma}\label{Lemma:F_Diff_Lower_Bound}
Consider the iterates $\mathbf{h}^k$ generated by Algorithm \ref{alg:CLA} and the cross-learning function $F$ defined in \eqref{eq:def_F}. If Assumptions \ref{As:Max_Grad_Deterministic}-\ref{As:DifferenceGradients} hold then, for all iterations before the stopping time condition \eqref{eq:stochasticalphafosp} is met, i.e. with $k<K$,  the expected difference in the function $F$ between two subsequent iterations, $\textbf{h}^k$ and $\textbf{h}^{k+1}$ is lower bounded by
\begin{align}
\EX[F(\textbf{h}^{k+1})-F(\textbf{h}^{k})]\geq \gamma_F, \quad k<K,
\end{align}
where the constant $\gamma_F$ is given by 
\begin{align}
&\gamma_F=\alpha-3(N+1)\beta^2L_{max}- \sqrt{N+1}\beta\textstyle\sum_{i=1}^N\frac{\sigma_{U_i}}{\sqrt{b_{U_i}}} \nonumber\\
&- \sqrt{N+1}\beta\textstyle\sum_{i=1}^N \mu_i \nonumber- (3 \eta^2 L_{max}+\eta)(\textstyle\sum_{i=1}^N B_{U_i}^2)\nonumber\\
&-\eta\textstyle\sum_{i=1}^N \mu_i(\textstyle\sum_{i=1}^N \frac{\sigma_{U_i}}{\sqrt{b_{U_i}}})-\eta(\textstyle\sum_{i=1}^N \frac{\sigma_{U_i}}{\sqrt{b_{U_i}}})^2\nonumber\\
&-\sqrt{N+1}\beta(\textstyle\sum_{i=1}^N \frac{\sigma_{U_i}}{\sqrt{b_{U_i}}}),
\end{align}
where $\eta\triangleq\sup_k \eta_k$ and $L_{\max}\triangleq\max_i L_i$.
\end{lemma}
\begin{proof}
See Appendix \ref{app:Lemma:F_Diff_Lower_Bound}. 
\end{proof}

The previous lemma shows that by using the iterates generated by Algorithm \ref{alg:CLA}, the cross-learning function $F$ is expected to increase at each step. Notice that, as the constant $\gamma_F$ (which determines the lower bound) depends on the step size of the gradient step $\eta^k$, the CDKOMP compression budget $\beta$ and the and the batch size $b_{U_i}$, it can be arbitrarily selected in order to ensure $\gamma_F>0$. Next, we use the result of Lemma \ref{Lemma:F_Diff_Lower_Bound} to show that the expected number of iterations $K$ needed to meet the stopping condition given by equation \eqref{eq:stochasticalphafosp} is necessarily upper bounded.
\begin{lemma}\label{Lemma:UpperBoundIterations}
Consider the iterates $\mathbf{h}^K$ generated by Algorithm \ref{alg:CLA} when satisfying the stopping condition given by \eqref{eq:stochasticalphafosp}, i.e., $\langle \hat{\nabla}_{h_i} U_i (h_i^K), (h-h_i^K)\rangle \leq \alpha$. Under Assumptions \ref{As:Max_Grad_Deterministic}-\ref{As:DifferenceGradients}, given initial policies $\textbf{h}^0$, the expected number of the iterations $K$ needed to reach the stopping condition \eqref{eq:stochasticalphafosp} is upper bounded by
\begin{align}
	\EX[K] \leq \frac{\EX[F(\textbf{h}^K)-F(\textbf{h}^0)]}{\gamma_F}.
\end{align}
where the constant $\gamma_F$ is given by Lemma \ref{Lemma:F_Diff_Lower_Bound}.
\end{lemma}
\begin{proof}
See Appendix \ref{app:Lemma:UpperBoundIterations}.
\end{proof}

By virtue of Lemma \ref{Lemma:UpperBoundIterations},  Algorithm \ref{alg:CLA} is guaranteed to stop after a finite number of iterations. By definition,  when the algorithm stops at iterate $K$, the resulting policies $\mathbf{h}^K$  meet the conditions in \eqref{eq:stochasticalphafosp}. Recall that these conditions are given with respect to the stochastic unbiased estimate of the gradient, $\hat{\nabla}_{h_i} U_i$. In order to guarantee that the policies  $\mathbf{h}^K$ are an $\alpha$-FOSP of problem  \eqref{eq:originalProblemCrosslearning} we need to guarantee the conditions of Definition \ref{def:alphafosp}, which involves the deterministic gradient $\nabla_{h_i} U_i$. In Theorem \ref{theo_HighProb} we state  that these conditions are guaranteed with \emph{high probability}.

\begin{theorem} \label{theo_HighProb}
Consider the iterates generated by Algorithm \ref{alg:CLA}. Under Assumptions \ref{As:Max_Grad_Deterministic}-\ref{As:DifferenceGradients}, the set of policies $\mathbf{h}^K$ generated at the stopping time is an $\alpha$-FOSP of problem \eqref{eq:originalProblemCrosslearning} with high probability. Namely, $\forall h \in \calC$
	\begin{align}\label{eq:final_eq}
	\Pr\left( \max_{i \in \{1,\ldots,N \} } \langle \nabla_{h_i} U_i (h_i^{K}), (h-h_i^{K}) \rangle \leq \alpha' \right) 
	\geq 1-\delta
	\end{align}
	where $\alpha'=\alpha + 2\epsilon\max_{i \in \{1,\ldots,N \}}  \sigma_{U_i} \sqrt{\frac{1}{\delta b_{U_i}}}$.
\end{theorem}
\begin{proof}
See Appendix \ref{app:theo_HighProb}.
\end{proof}

This result guarantees the convergence of Algorithm \ref{alg:CLA}. There are several elements to take into consideration regarding the point at which it converges. The approximation $\alpha'$ deviates from $\alpha$ by the proximity factor $\epsilon$ of cross-learning, the high probability guarantee $1-\delta$,  and the bound given by Assumption \ref{As:DifferenceGradients}, which can be made arbitrarily small by increasing the batch size $b_{U_i}$.
 

%% file: 07_numerical_results.tex
\section{Numerical Results}\label{sec:NumericalResults}

\begin{figure}[t]
	\centering
    \includegraphics[scale=1]{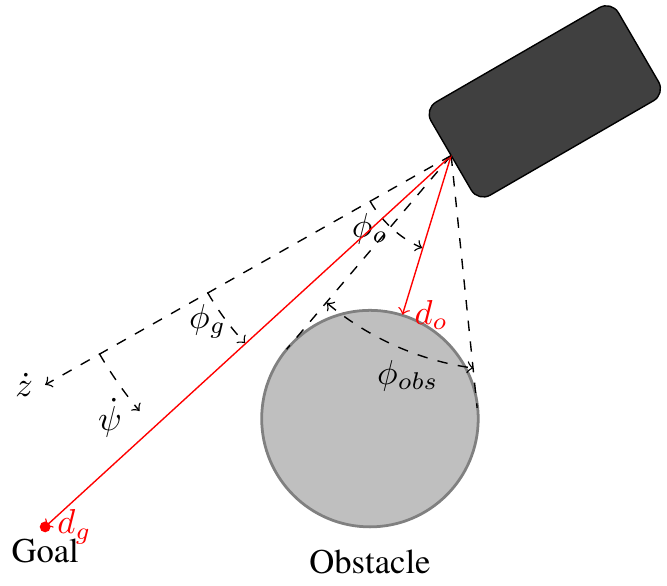}        		
	\caption{Nonholonomic robot model used for simulations. The system state is given by the tuple $(d_o,\phi_o,d_g,\phi_g,\phi_{obs})$; where $d_o$ and $d_g$ are the distances between the agent and the obstacle and goal, respectively; $\phi_o$ and $\phi_g$ are the angles between the agent and the obstacle and the goal; and $\phi_{obs}$ is the angle of occlusion of the obstacle. The system actions are given by $(\dot z,\dot \psi)$ where $\dot z$ is the radial velocity and $\dot \psi$ is the angular velocity of the agent.}
	\label{fig:robot_model}
\end{figure}

To evaluate the performance of our proposed cross-learning methodology, we resort to a continuous navigation problem with obstacles. We consider the  robot model shown in Figure \ref{fig:robot_model}. Following this model, the state of the agent is given by $s=(d_o,\phi_o,d_g,\phi_g,\phi_{obs})$; where $d_o$ and $d_g$ are the distances between the agent and the obstacle and goal, respectively; $\phi_o$ and $\phi_g$ are the angles between the agent and the obstacle and the goal; and $\phi_{obs}$ is the angle of occlusion of the obstacle. Real-valued actions  $(\dot z,\dot \psi)$ represent  the radial and angular velocity of the agent, respectively. Furthermore, we only consider positive $\dot z$ actions. The system MDP is formed by the  following dynamics
\begin{align}
\begin{split}
x_{t+1}&=x_t+T_s \dot{z}_t \cos(\psi_t), \\ 
y_{t+1}&=y_t+T_s \dot{z}_t \sin(\psi_t), \\ 
\psi_{t+1}&=\psi_t+T_s \dot{\psi_t},
\end{split}
\end{align}
where $T_s \in \real^+$ is the time step and both $x$ and $y$ are the cardinal points of the agent's state measured from the point $(0,0)$. Furthermore, for each task, the agent takes actions from a Gaussian distribution with covariance $\Sigma=\diag(0.05,0.05)$. The mean of this Gaussian distribution is given by a policy formed by the linear combination of Gaussian kernels which are non-symmetric, and their covariance matrix is given by $\Sigma_{\kappa}=\diag(1, \pi/5, 1, \pi/5, \pi/10)$. The tasks over which the agent will train correspond to navigating environments with different types of obstacles. For all tasks, a collision to an obstacle results in a negative reward $r(s,a)=-100$. Otherwise  the reward is mostly given by the distance to goal, specifically  $r(s,a)=10-10||d_g||_{2}$. The starting point of the agent is randomly drawn from $[0 ,10]\times [0,10]$ in the absolute $x \times y$ axes of the environment. Furthermore, the time step in the MDP dynamics is selected to be $T_s=0.5$. Other parameters of the system are a discount factor of $\gamma=0.9$, a fixed step size of $\eta=0.1$. A batch version of the stochastic gradient in Algorithm \ref{alg:unbiased} is computed averaging $4$ samples  per gradient ascent step. Furthermore, to avoid memory explosion an upper bound in the model order is set at $400$ kernels.

\subsection{Navigation with a Single Obstacle}
\label{subsec:NavigationSingle}

In order to study the problem of learning multiple tasks, we consider three  scenarios with obstacles that differ in positions and radius. Specifically the obstacles are centered at $c_1=(7,2)$  with radius $r_1=0.5$ for the first task, $c_2=(2,2)$ with $r_2=1$ for the second task and $c_3=(7,7)$ with $r_3=2$ for the third task. We consider five policies which are trained in these scenarios. First, we consider the three policies agnostically trained for each specific task, without considering the data collected by the other agents. Then, we also consider an average policy over all tasks (equivalent to a cross-learning policy with $\epsilon=0$, and a cross-learning policy with $\epsilon=3$.

\begin{figure}[t]
    \centering
    \includegraphics[scale=1]{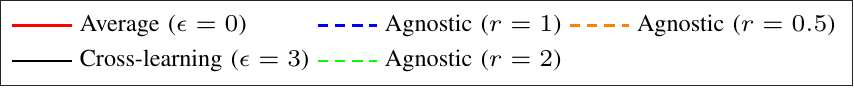}    
	\hspace{-0.2cm}    
    \subfigure[Circle obstacle.]{
    \includegraphics[scale=1]{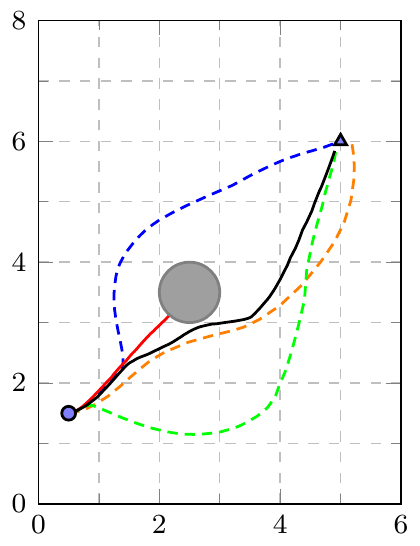}     
    \label{fig:robot_circle}    
	 }
	 \hspace{-0.6cm}
    \subfigure[Ellipse obstacle.]{
    \includegraphics[scale=1]{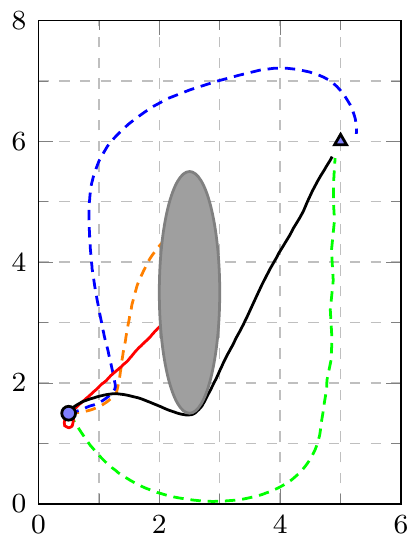}         
    \label{fig:robot_ellipse}
	 }     
	\caption{Sample trajectories for the policies obtained after $29{,}000$ training iterations. An environment previously seen during training (circle) and an unseen environment (ellipse) have been chosen for evaluation. Agents trained for a larger radius tend to avoid the obstacles in a more conservative way, while only the cross-learned policy ($\epsilon=3$), does so in an efficient manner.} 
    \label{fig:single_obstacle_nav}
\end{figure}

For evaluation we use the environments shown in Fig. \ref{fig:single_obstacle_nav}. We sample trajectories to illustrate the behavior of each policy when facing an obstacle with the starting point set at $s_{\texttt{start}}=[0.5,1.5]$ and the goal at $s_{\texttt{goal}}=[5,6]$. Figure \ref{fig:robot_circle} shows a circular obstacle centered at $(x_{\texttt{obs}},y_{\texttt{obs}})=(2.5,3.5)$. Figure \ref{fig:robot_ellipse} shows an ellipse obstacle also centered at $(x_{\texttt{obs}},y_{\texttt{obs}})=(2.5,3.5)$, but elongated along the $y$-axis. Specifically, the ellipse is defined by $((x-x_{\texttt{obs}})/\alpha_x)^2+((y-y_{\texttt{obs}})/\alpha_y)^2\leq 1$, with $(\alpha_x,\alpha_y)=(0.5,2)$.

Overall, all policies are capable of avoiding the circular obstacle in Fig. \ref{fig:robot_circle}. However, the policies trained for a larger radius do so in a conservative way, the more conservative the larger their training radius. Furthermore, as expected, the policy trained for the same radius $r=0.5$ achieves the best performance among those trained by the agnostic agents. However the cross-learned policy with $\epsilon=3$ obtains the best overall performance. This is credited to training on a larger dataset, given the ability of the cross-learning algorithm to combine the data collected by all agents during the process. In contrast, the average policy, given by $\epsilon=0$ crashes against the obstacle. In the case of the ellipse obstacle shown in Fig. \ref{fig:robot_ellipse},  not only the average policy ($\epsilon=0$) crashes, but also the one trained for a smaller circular radius of $r=0.5$, as it is uncapable of steering around the ellipse. The other agnostic agents trained for larger radius ($r=1$ and $r=2$) reach the goal but in an excessively conservative way. Similar to the circle case, the policy trained for a larger radius ($r=2$) does so in a more conservative manner.

\begin{figure}[t]
    \centering
    \subfigure[Environment with obstacle of radius $r=0.5$.]{
    \includegraphics[scale=1]{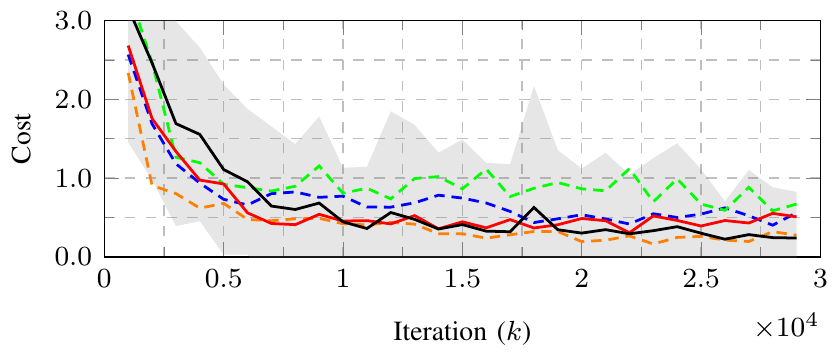}    	
    \label{fig:rew_multi_0}    
	 }
    \subfigure[Environment with obstacle of radius $r=1$.]{
    \includegraphics[scale=1]{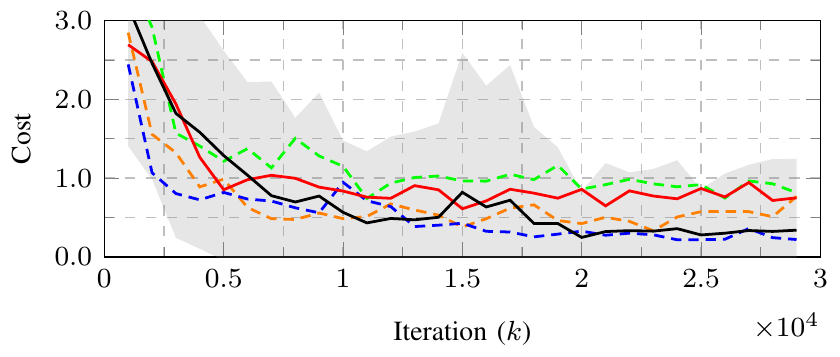}    	    
    \label{fig:rew_multi_1}
	 }    
    \subfigure[Environment with obstacle of radius $r=2$.]{
    \includegraphics[scale=1]{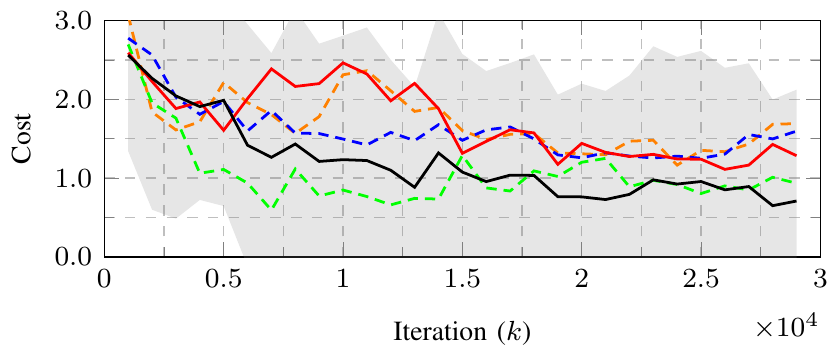}    	    
    \label{fig:rew_multi_2}
	 }     	  
     \includegraphics[scale=1]{pdf/navigation_legend.pdf}    	 
	\caption{Evolution of the absolute value of the average cost with respect to the training iteration. The scenario for evaluation comprises a circular obstacle centered at  $(x_{\texttt{obs}},y_{\texttt{obs}})=(2,2)$. The plots represent the average cost over $500$ independent trials, where the initial state is randomly drawn. The standard deviation of the cross-learning policy ($\epsilon=3$) is shown by the shaded area.}
	  \label{fig:rew_multi}
\end{figure}

\begin{remark}
Note that setting $\epsilon=0$ forces consensus on the policies, and this is  too restrictive. Alternatively, with $\epsilon>0$, each agent trains a different policy that adapts better to the data that it is collecting, and in this example, to adapt to a different radius. Policies are bound together by the constraint, so that the data collected by all agents are combined by the cross-learning algorithm. This data combination allows an agent to easily adapt to an object not seen during training. However, if $\epsilon$ is too large, the cross-learning constraint becomes inactive and the agents are agnostic of each other and the common policy $g$ becomes undetermined (cf. \eqref{eq:projectionProblemC}).    
\end{remark}

A more qualitative look at the behavior of the different policies can be gained by looking at the cost obtained during training. This is shown in Figure \ref{fig:rew_multi}, where we have plotted the absolute value of the average cost over the training iterations, with cost=-reward. The average cost has been obtained after $500$ independent trials. We plot three different figures, one for each of the training environments, $r=0.5$, $r=1$, and $r=2$. First, notice that out of the agnostic policies, the one specifically trained for its respective environment outperforms all the rest. For example, in Fig. \ref{fig:rew_multi_0}, corresponding to the scenario with a radius $r=0.5$, the agnostic policy trained for $r=0.5$, outperforms the rest. This behavior is repeated in Fig. \ref{fig:rew_multi_1} and \ref{fig:rew_multi_2}. However, agnostic policies do not perform well in different scenarios to the one being trained on. For example, the policy trained for $r=0.5$, performs the worst for the scenario with radius $r=2$ (Fig. \ref{fig:rew_multi_1}). This is validated by our previous findings in Fig. \ref{fig:single_obstacle_nav}, as the policy trained with radius $r=0.5$ has trouble maneuvering around bigger obstacles. More importantly, the cross-learned policy with $\epsilon=3$ performs comparably well to the task-specific policy in each case, even performing better in the case of $r=2$ (Fig. \ref{fig:rew_multi_2}). In contrast, as observed before, the average policy ($\epsilon=0$) underperforms for all the scenarios.

\begin{figure}[t]
	\centering
    \includegraphics[scale=1]{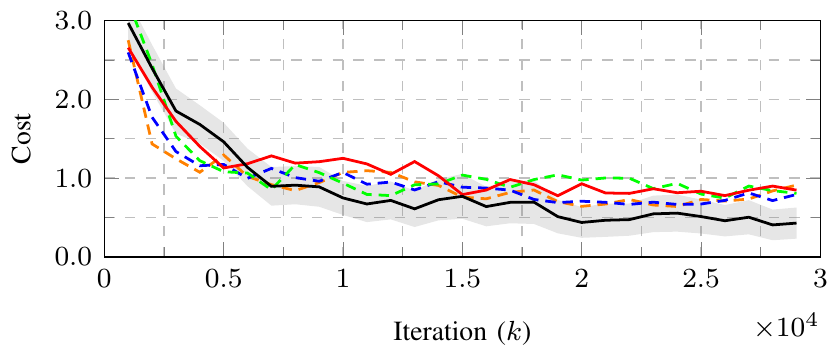}    	
	\caption{Absolute value of the cost averaged over all task vs. training iteration. The scenario for evaluation consists of the average cost with respect to the three scenarios $r=0.5$, $r=1$, and $r=2$. The plots represent the average cost over $500$ independent trials, where the initial state is randomly drawn. The standard deviation of the cross-learning policy ($\epsilon=3$) is shown by the shaded area.}
	\label{fig:rew_comparison}
\end{figure}

The adaptability of the cross-learned policy can be seen more clearly in Figure \ref{fig:rew_comparison}. In this figure, we plot the absolute value of the average cost over the three tasks. Overall, any of the agnostically trained policies, if averaged over tasks, performs similarly. This is also the case of the average policy ($\epsilon=0$). In contrast, the cross-learned policy ($\epsilon=3$)  outperforms the rest of the policies.

\subsection{Navigation with Multiple Obstacles}
\label{subsec:NavigationMultiple}

Recall that the robot model introduced in Fig. \ref{fig:robot_model} operates on a local frame. Thus, we can use the training on an individual obstacle and goal to study the generalization capabilities of our proposed method to multiple goals and obstacles of different shapes and sizes. To this end, we consider the scenario shown in Figure \ref{fig:robot_multi}. The same policies trained in the environment described in the previous section are evaluated in this new environment. The agents start at $[0.5,1.5]$ and navigate to a series of subsequent objectives in the following order $\texttt{goal}_1=[5,6]$, $\texttt{goal}_2=[5.5,1.5]$, and $\texttt{goal}_3=[11,5]$. The state of the agent is determined by the current goal and the closest obstacle.

First, notice that as in the single goal and obstacle navigation case of Fig. \ref{fig:single_obstacle_nav}, the average policy ($\epsilon=0$) fails to navigate, crashing against the first obstacle on its way to the first goal. Compared to this, the agnostically trained policies perform slightly better. The agnostic policy for $r=0.5$  performs well when avoiding the obstacles with similar small radius, but then crashes to the larger ellipse. The policy trained for the larger radius ($r=2$) avoids obstacles in a more conservative way, which ultimately prevents it from navigating to the last goal, overshooting and then crashing into the ellipse as well. However, the policy trained with the obstacle of radius $r=1$ succeeds in navigating to through all the goal, even though it does so in a conservative way. Compared to all of the agnostically trained policies, the cross-learned policy ($\epsilon=3$) navigates faster through the environment, achieving all the goals and avoiding all the obstacles.

\begin{figure}[t]
	\centering
    \includegraphics[scale=1]{pdf/navigation_legend.pdf}    	
    \includegraphics[scale=1]{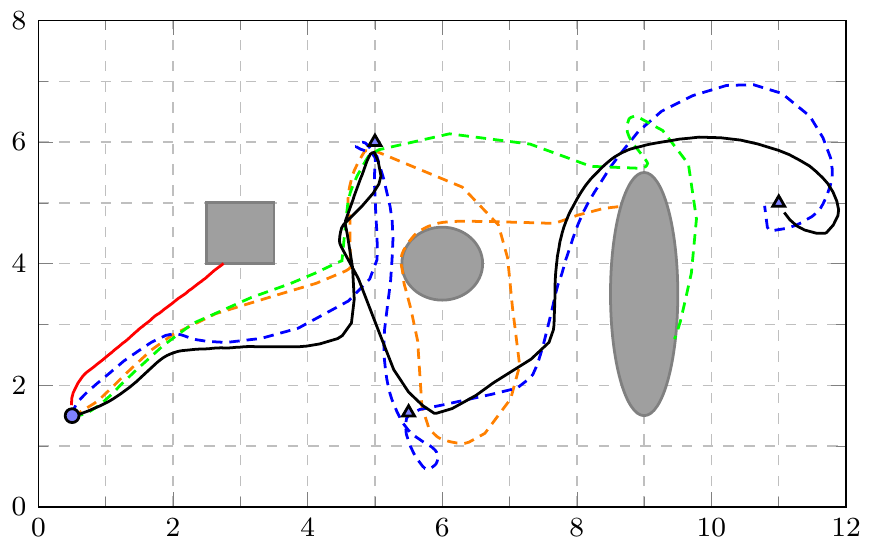}    		
	\caption{Navigation with multiple obstacles. Sample trajectories of the policies obtained after $29{,}000$ training iterations. The agents start at $[0.5,1.5]$ and navigate to a series of subsequent objectives in the following order $\texttt{goal}_1=[5,6]$, $\texttt{goal}_2=[5.5,1.5]$, and $\texttt{goal}_3=[11,5]$. The state of the agent is determined by the current goal and the closest obstacle.}
	\label{fig:robot_multi}
\end{figure}	

%% file: 08_conclusions.tex
\section{Conclusions}
\label{sec:Conclusions}

We have introduced cross-learning as a framework for training policies in multiple tasks jointly. Posed as a constrained optimization problem, the learning strategy  maximizes  task specific  rewards while  coupling the policies together. The cross-learning policies live in between agnostic policies that train for their individual tasks separately, and   consensus policies that train for  average rewards. In between these two extremes, the policies gain the capability of learning to maximize separate rewards  but sharing training  data acquired across tasks.  This capability makes the cross-learning policies successful at two novel aspects.  It produces a central policy   that can be used as starting point to adapt quickly to one of the tasks trained for, in a situation when the agent does not know which task is currently facing. Moreover, this central policy generalizes better to tasks related but different to those seen during training. These improvements are obtained through a projected stochastic gradient ascent algorithm that incorporates state and reward data   sequentially, converging with high probability to a near optimal policy.   We have studied the properties of the required projection step, finding a relaxed formulation with a closed-form solution, useful when the computational resources are limited, and added an extra projection to curtail the number of kernels  in order to avoid memory explosion. Numerical results over  navigation tasks have shown that  policies can avoid obstacles not trained for, and can manage through environments with obstacles of multiple shapes.

%% file: 09_appendix.tex
\section{Common Dictionary Kernel Orthogonal Matching Pursuit (CDKOMP)} \label{app:cdkomp}

In this appendix, we cover the required procedure to reduce the policies to a lower dimensional space. Starting by setting all the $N$ policies at iteration $k=0$ to zero, i.e. $h^0_i= 0$, each stochastic gradient ascent iteration results in new policies
\begin{align} \label{eq:GradAscent} 
	\bar{h}_{i}^{k}=h_i^k+\eta^k \hat{\nabla}_{h_i} U_i (h_i^k), \quad i=1,\dots,N,
\end{align}
where $\eta^k \in (0,1)$ is the step size and $\hat\nabla_{h_i}U_i(h_i^k)$ are obtained by the unbiased estimate for the gradient of $U_i$ given by Algorithm \ref{alg:unbiased}. For the matter of Algorithm \ref{alg:MDKOMP}, cross-learning policy $g$ can be considered as one policy whose gradient is always $0$, namely $\bar{g}^k=g^k$. Herein, cross-learning policy $g$ will be treated as any other policy.

 Policies $h_i$ can be represented by their respective dictionaries by rewriting $\eta^k \hat\nabla_{h_i}U_i(h_i^k,\cdot)=w_{i,k} \kappa(s_{i,k},\cdot)$. Hence, resulting at the $k$-th iteration in the dictionary $\textbf{D}_{h_i^k} = \{ s_{i,1},\dots,s_{i,M_i} \}$ and weights $\textbf{W}_{h_i^k}=\{w_{i,1},\dots,w_{i,M_i}\}$. A priori, each policy $h_i$ may have a different amount of kernels in their dictionaries, meaning that $M_i \neq M_j$, for $i\neq j$. However, this cannot happen because of Assumption \ref{assumption:kernels}, which states that all dictionaries should only differ in one kernel, which is the last kernel obtained via the unbiased estimate Algorithm \ref{alg:unbiased}. Thus, at the $k$-th iteration, all $\{h_i^k\},g^k$  policies must have $L^k$ kernels and the policy $\bar h_i^{k}$ produced after the gradient ascent step, $L^k+1$ kernels. Notice that the gradient ascent step adds no kernel to $g^k$. Hence, when reducing the complexity of the dictionary, an element is either eliminated from all dictionaries or not eliminated from any of them.

After obtaining an unbiased estimate for the gradient of $U_i$ via Algorithm \ref{alg:unbiased} and computing the gradient ascent, we perform the cross-learning projection \eqref{eq:projectionProblemC} or its relaxed version \eqref{eq:projectionProblemR} and obtain the next policy, 
\begin{align}\label{eq:app_projection}
(\{\tilde h_i^{k+1}\},\tilde{g}^{k+1})=\calP_{\calC} \left[\{ h_i^k+\eta^k \hat\nabla_{h_i}U_i(h_i^k) \},g^k\right].
\end{align}
In this case, after the projection, all policies $\tilde h_i^{k+1}$ and $\tilde{g}^{k+1}$ will have the same number of kernels $L_{k+1}=L_{k}+N$. If no kernel is eliminated from the dictionary, at time $k$, the $N+1$ policies will each have $k N$ number of kernels. This ultimately can result in computational problems as $k$ increases. Thus, we aim to find a way to reduce the number of kernels while maintaining their function approximation capabilities. Hence, we compute the approximation error of eliminating each element in the dictionary. Abusing notation, we omit the step $k$ in the notation, as what follows is independent from the iteration step and we will refer to the cross-learning policy $h_{N+1}=g$. Defining $\tilde  h_i$ by dictionary $\textbf{D}_{\tilde h_i}=\{s_{i,1},\dots,s_{i,L}\}$ and the vector of stacked weights $\textbf{w}_{\tilde h_i}=[w_{i,1}^T,\dots,w_{i,L}^T]^T$, we refer to $\tilde h_i^j$ with $j \in \{1,\dots,L\}$ as the policy that has the same kernel centers as $\textbf{D}_{\tilde h_i}$ but not element $s_j$. Then, we can compute the minimum distance between $h_i$ and $h_i^j$  as, follows
\begin{align}	\label{eq:ErrorDict}
	e(j,\tilde h_{i})&=\underset{\mathbf{w}_{\tilde h_i^j} \in\real^{L \times P}} {\min}||\tilde h_i-\tilde h_i^j||^2_{\calH}\\
	&=\mathbf{w}_{\tilde h_i}^T \mathbf{K}_{\textbf{D}_{\tilde h_i}, \textbf{D}_{\tilde h_i} }  \mathbf{w}_{\tilde h_i} \nonumber\\
	&+\underset{\mathbf{w}_{\tilde h_i^j} \in\real^{L \times P}} {\min}
	\mathbf{w}_{\tilde h_i^j}^T \mathbf{K}_{\textbf{D}_{\tilde h^j_i}, \textbf{D}_{\tilde h^j_i} }  \mathbf{w}_{\tilde h^j_i} 
	- 2 \mathbf{w}_{\tilde h_i^j}^T \mathbf{K}_{\textbf{D}_{\tilde h^j_i}, \textbf{D}_{\tilde h_i} }  \mathbf{w}_{\tilde h_i} 
	\nonumber
\end{align}
by substituting $\tilde h_i$ and $\tilde h_i^j$ by their weights and dictionary representation and defining the block matrices $\mathbf{K}_{(\cdot,\cdot)}$ whose $(l,m)$-th blocks of size $p \times p$ are $\Kappa (s_l,s_m)$, with $s_l ,s_m \in \textbf{D}_{\tilde h_i} $ in $\mathbf{K}_{\textbf{D}_{\tilde h_i}, \textbf{D}_{\tilde h_i} }
$, $s_l \in \textbf{D}_{\tilde h_i^j},  s_m \in \textbf{D}_{\tilde h_i}$ in $\mathbf{K}_{\textbf{D}_{\tilde h^j_i}, \textbf{D}_{\tilde h_i} }
$ and $s_l ,s_m \in \textbf{D}_{\tilde h^j_i}$ in $\mathbf{K}_{\textbf{D}_{\tilde h^j_i}, \textbf{D}_{\tilde h^j_i} }
$. Notice that element $\mathbf{w}_{\tilde h_i}^T \mathbf{K}_{\textbf{D}_{\tilde h_i}, \textbf{D}_{\tilde h_i} }  \mathbf{w}_{\tilde h_i} $ is constant in the minimization, as it does not depend of $\mathbf{w}_{\tilde h^j_i}$. Furthermore, the previous problem is a least-squares which close form solution can be given by 
\begin{align}
\mathbf{w}_{\tilde h_i^j}^\star=
\mathbf{K}_{\textbf{D}_{\tilde h^j_i}, \textbf{D}_{\tilde h_i^j} }^\dagger
\mathbf{K}_{\textbf{D}_{\tilde h^j_i}, \textbf{D}_{\tilde h_i} }
\mathbf{w}_{\tilde h_i},
\label{eq:MinimalErrorWeights}
\end{align}
where $(\cdot)^{\dagger}$ denotes the Moore-Penrose pseudo-inverse. After arguing that the policies share the kernel centers, and ergo share their dictionaries $\textbf{D}_{\tilde h_i}$, we have obtained a closed-form solution to the error of taking one element out of the dictionary. Now, we can evaluate which the kernel eliminations renders the smallest error in all of the policies. First, we must compute the largest error per element in all of the policies,
\begin{align} \label{problem:max_error}
\tilde e(j)=\underset{i \in \{1,\dots,N\}}{\max} \quad e(j,\tilde h_i) , 
\end{align}
where $e(j,\tilde h_i)$ is the error obtained by equation \eqref{eq:ErrorDict}. Hence, $\tilde e(j)$ denotes the maximum error in all policies if the $j$-th kernel were to be eliminated. This means that if $\tilde e(j)$ is small, the error of eliminating kernel $j$ in all policies is small. Now, we need to compute which element renders the smallest error. Namely,
\begin{align}
\tilde e^\star=\underset{j \in \{1,\dots,L\}}{\min} \quad \tilde e(j). 
\end{align}
\begin{algorithm}[t]
	\label{alg:KOMP}
	\caption{Common Dictionary Kernel Orthogonal Matching Pursuit (CDKOMP) }
	\label{alg:MDKOMP}
	\textbf{Input:} Functions $\tilde h_i^{k+1}, i=1,\dots,N+1$ defined by common dictionary $\textbf{D}_{\tilde h_i} \in \R^{q\times L_{k+1}}$ and weights $\textbf{w}_{\tilde h_i} \in \R^{p  L_{k+1}}$. Compression budget $\beta>0$. 
	\begin{algorithmic}[1]
		\State \textit{Initialize:} Error $\tilde e^\star=0$.
		\While {$\tilde e^\star < 0$ and $0<M_{k+1}$}  
		\For {$i=1,\dots ,N+1$ and $j=1, \dots , L_{k+1}$} 
		\State Compute minimal error $e(j,\tilde h_{i})$ with \eqref{eq:ErrorDict}
		\EndFor
		\State Compute the maximum error per kernel  $\tilde e(j)$ with \eqref{problem:max_error} 
		\State Compute less informative element $ j^\star= \argmin_j \tilde e(j)$
		\State Compute pruning error $e^\star = \min_j \tilde e(j)$
		\If {$e^\star< \beta$}
		\State Update weight $\textbf{w}_{\tilde h_i} \leftarrow  \textbf{w}_{\tilde h_i^j}^\star$ with \eqref{eq:MinimalErrorWeights}
		\State Prune dictionary $\textbf{D}_{\tilde h_i} \leftarrow \textbf{D}_{\tilde h_i^j}$ 
		\State Decrease model order $L_{k+1} \leftarrow  L_{k+1}-1$
		\EndIf
		\EndWhile
		\\ \Return dictionary $\textbf{D}_{\tilde h_i} \in \R^{q\times L_{k+1}}$, weights $\textbf{w}_{\tilde h_i} \in \R^{p  L_{k+1}}$	\end{algorithmic}
\end{algorithm}
As mentioned earlier, our goal is to eliminate an element from the dictionary, as long as the error that we incur is upper bounded by a value $\beta$. Hence, if $\tilde e^\star<\beta$, element $j^\star$, being the one that verifies $\tilde e^\star=\tilde e(j^\star)$ is removed. The final result of this pruning process are $N$ functions $h^{k+1}_i$ that are less than $\beta$ close to the functions before the pruning procedure and are represented by the same or less elements. This results in the Common Dictionary Kernel Orthogonal Matching Pursuit (CDKOMP) procedure, a summarized description of which is shown in Algorithm \ref{alg:MDKOMP}. This process can be equivalently interpreted as a projection to a RKHS of lower dimension.  If we denote by $\textbf{D}_{\tilde h_i}$ the dictionary of the output of Algorithm \eqref{alg:MDKOMP}, then the resulting policies $h_i^{k+1}$ can be interpreted as, the projection
\begin{equation}\label{eq:DictionaryPruning}
(\{h^{k+1}_{i}\},g^{k+1})=\calP_{\calH_{\bDi}} \left [ \calP_{\calC} \Big[\{h_i^{k} +\eta^k \hat{\nabla}_{h_i} U_i (h_i^{k})\}, g^k  \Big] \right].
\end{equation}
Now, after showing that the pruning  procedure can be seen as a projection over a lower dimension space, we show that this projection is equivalent to running a biased projected gradient ascent. 

\begin{proposition}\label{prop:prunning}
	The projection to pruned dictionary step in Algorithm \ref{alg:CLA} is equivalent to running biased projected gradient ascent, with bias
\begin{align} \label{eq:bias_def}
\{b^k_i\}=\calP_{\calH_{\bDi}}  &\left [ \calP_{\calC} \Big[\{h^k_i +\eta^k \hat{\nabla}_{h_i} U_i (h^k_i)\},g^k \Big] \right]  \nonumber\\
&- \calP_{\calC} \Big[ \{h_i^k + \eta^k \hat{\nabla}_{h_i} U_i (h^k_i) \},g^k \Big] 
\end{align}
bounded by the compression budget, i.e., $\|b_i^k\|\leq\beta$ for all $k>0,i=1,\dots,N+1$.
\end{proposition}
\begin{proof}
	Note that following the procedure in Algorithm \ref{alg:CLA}, the projection to the dictionary of reduced dimension is given by \eqref{eq:DictionaryPruning}. Adding and subtracting the argument of this projection, that is $\calP_{\calC} \Big[\{h^k_i + \eta^k \hat{\nabla}_{h_i} U_i  (h^k_i)\},g^k \Big]$, it is possible to write the resulting policies $h_i^{k+1}$ as
	\begin{align}
	(\{h_{i}^{k+1}\},g^{k+1})=\calP_{\calH_{\bDi}} \left [ \calP_{\calC} \Big[ \{h^k_i +\eta^k \hat{\nabla}_{h_i} U_i  (h^k_i) \},g^k\Big] \right] \nonumber \\ 
	- \calP_{\calC} \Big[\{h^k_i +\eta^k \hat{\nabla}_{h_i} U_i  (h^k_i) \Big\},g^k]\nonumber \\
	+ \calP_{\calC} \Big[ \{h^k_i +\eta^k \hat{\nabla}_{h_i} U_i  (h^k_i)\},g^k \Big].
	\end{align}
Then, simply using the definition of the bias $b^k_i$, given by equation \eqref{eq:bias_def}, the previous expression can we written as
		\begin{align}
	(\{h^{k+1}_i\},g^{k+1})=\calP_{\calC} \Big[ \{h^k_i + \eta^k \hat{\nabla}_{h_i} U_i  (h^k)\},g^k \Big] + \{b^k_i\}.
	\end{align}
	Using the compression budget imposed in Algorithm \ref{alg:MDKOMP}, it follows that $|b^k_i| \leq \beta$ for all $k>0,i=1,\dots,N+1$.
\end{proof}

\section{Proof of Lemma \ref{Lemma:F_Diff_Lower_Bound}} \label{app:Lemma:F_Diff_Lower_Bound}
\begin{proof}
Let us define the difference between the policy vector $\textbf{h}$ at iteration $k$ and $k+1$ by $\Delta \textbf{h} ^k=\textbf{h}^{k+1}-\textbf{h}^k=[h^{k+1}_1-h^{k}_1,\dots,h^{k+1}_N-h^{k}_N]$. Let us also define the internal product between two policy vectors $\langle \textbf{h}^{k+1}, \textbf{h}^{k} \rangle=\sum_{i=1}^N \langle h^{k+1}_i,h^{k}_i\rangle_{\calH}$, which induces the norm $||\textbf{h}||=\sqrt{\sum_{i=1}^N ||h_i ||_{\calH}^2}$.

Applying the Mean Value Theorem in each policy $h_i$, for $f_i=\lambda_i h^{k+1}_i+(1-\lambda_i)h_i^k$, with $\lambda_i \in [0,1]$ and $ i=1,\dots,N$, we have that
\begin{align}\label{eq:MeanValueTheorem}
F(\textbf{h}^{k+1})-F(\textbf{h}^{k})=\langle \nabla_{\textbf{h}} F (\mathbf{f}), \Delta \textbf{h}^k \rangle,
\end{align}
where $\textbf{f}=[f_1,\dots,f_N]$. Adding and subtracting the stochastic gradient of equation \eqref{eq:MeanValueTheorem}, i.e., $\langle \hat{\nabla}_{\textbf{h}} F (\mathbf{f}), \Delta \textbf{h}^k \rangle$ we have
\begin{align}
F&(\textbf{h}^{k+1})-F(\textbf{h}^{k})\nonumber\\
&=\langle \hat{\nabla}_{\textbf{h}} F(\textbf{h}^k), \Delta \textbf{h}^k \rangle +\langle \nabla_{\textbf{h}} F(\textbf{f}) - \hat{\nabla}_{\textbf{h}} F(\textbf{h}^k) , \Delta \textbf{h}^k \rangle  \\
&=\textstyle\sum_{i=1}^N\langle \hat{\nabla}_{h_i} U_i(h_i^k), \Delta h_i^k \rangle +\langle \nabla_{\textbf{h}} F(\textbf{f}) - \hat{\nabla}_{\textbf{h}} F(\textbf{h}^k) , \Delta \textbf{h}^k \rangle. \nonumber
\end{align}
Using the fact that the stopping condition equation \eqref{eq:stochasticalphafosp} is not met for $k<K$, we can define $j$ as the index at which the maximum internal product $j=\argmax_i \langle \hat{\nabla}_{h_i} U_i(h_i^k), \Delta h_i^k \rangle$ is obtained. By virtue of the stopping condition \eqref{eq:stochasticalphafosp} and Cauchy-Schwartz inequality we obtain, 	
\begin{align}
F(&\textbf{h}^{k+1})-F(\textbf{h}^{k})\nonumber\\
&= \langle \hat{\nabla}_{h_j} U_j(h_j^k), \Delta h_j^k \rangle+\textstyle\sum_{i=1,i\neq j}^N\langle \hat{\nabla}_{h_i} U_i(h_i^k), \Delta h_i^k \rangle \nonumber \\ 
&+\langle \nabla_{\textbf{h}} F(\textbf{f}) - \hat{\nabla}_{\textbf{h}} F(\textbf{h}^k) , \Delta \textbf{h}^k \rangle  \\
&\geq \alpha +\textstyle\sum_{i=1,i\neq j}^N\langle \hat{\nabla}_{h_i} U_i(h_i^k), \Delta h_i^k \rangle\nonumber \\
&+\langle \nabla_{\textbf{h}} F(\textbf{f}) - \hat{\nabla}_{\textbf{h}} F(\textbf{h}^k) , \Delta \textbf{h}^k \rangle\\
&\geq \alpha-||\hat{\nabla}_{\textbf{h}}F(\textbf{h}^k)||_{\calF} ||\Delta \textbf{h}^k|| \nonumber\\
&-\langle \nabla_{\textbf{h}} F(\textbf{f})-\hat{\nabla}_{\textbf{h}} F(\textbf{h}^k) , \Delta \textbf{h}^k \rangle.
\end{align}
Observe that the previous inequality holds as adding the $j$-term to the policy vector increases its norm. Then, adding and subtracting the internal product between the deterministic gradient and the difference $\nabla \textbf{h}^k$, $\langle\nabla_{\textbf{h}}F(\textbf{h}^k),\Delta \textbf{h}^k \rangle$, the right hand side of the previous expression can be rewritten as, 
\begin{align}
F&(\textbf{h}^{k+1})-F(\textbf{h}^{k})\nonumber\\
&\geq \alpha-||\nabla_{\textbf{h}}F(\textbf{h}^k)|| ||\Delta \textbf{h}^k||- \langle \nabla_{\textbf{h}} F(\textbf{h}^k) - \hat{\nabla}_{\textbf{h}} F(\textbf{h}^k) , \Delta \textbf{h}^k \rangle\nonumber\\
&-\langle \nabla_{\textbf{h}} F(\textbf{f}) - \nabla_{\textbf{h}} F(\textbf{h}^k) , \Delta \textbf{h}^k \rangle.
\label{eq:Ineq_Last_with_SumAndSubstract}
\end{align}
Applying Cauchy-Schwartz inequality and explicitly rewriting the internal product between the cross-learning function as the sum of the term-wise internal products we obtain, 
\begin{align}
&F(\textbf{h}^{k+1})-F(\textbf{h}^{k})\nonumber\\
&\geq \alpha-||\hat{\nabla}_{\textbf{h}}F(\textbf{h}^k)|| ||\Delta \textbf{h}^k||-|| \nabla_{\textbf{h}} F(\textbf{h}^k) -  \hat{\nabla}_{\textbf{h}} F(\textbf{h}^k) || ||\Delta \textbf{h}^k ||\nonumber\\
&-\textstyle\sum_{i=1}^N \langle \nabla_{h_i} U_i(f_i) - \nabla_{h_i} U_i(h_i^k) , \Delta h_i^k \rangle.
\label{eq:Ineq_Spliting_IntProd}
\end{align}
Using Cauchy-Schwartz inequality and by virtue of Assumption \ref{As:Grad_Lipschitz} and upper bounding the Lipschitz constants by $L_{\max}=\max_i L_i$ the last term can be bounded by, 
\begin{align}
&F(\textbf{h}^{k+1})-F(\textbf{h}^{k})\nonumber\\
&\geq \alpha-||\hat{\nabla}_{\textbf{h}}F(\textbf{h}^k)|| ||\Delta \textbf{h}^k||-|| \nabla_{\textbf{h}} F(\textbf{h}^k) -  \hat{\nabla}_{\textbf{h}} F(\textbf{h}^k) || ||\Delta \textbf{h}^k ||\nonumber\\
&-\textstyle\sum_{i=1}^N || \hat{\nabla}_{h_i} U_i(f_i) - \nabla_{h_i} U_i(h_i^k) || || \Delta h_i^k ||\\
&\geq \alpha-||\hat{\nabla}_{\textbf{h}}F(\textbf{h}^k)|| ||\Delta \textbf{h}^k||-|| \nabla_{\textbf{h}} F(\textbf{h}^k) -  \hat{\nabla}_{\textbf{h}} F(\textbf{h}^k) || ||\Delta \textbf{h}^k ||\nonumber\\
&-\textstyle\sum_{i=1}^N L_i || \Delta h_i^k ||^2\\
&\geq \alpha-||\hat{\nabla}_{\textbf{h}}F(\textbf{h}^k)|| ||\Delta \textbf{h}^k||-|| \nabla_{\textbf{h}} F(\textbf{h}^k) -  \hat{\nabla}_{\textbf{h}} F(\textbf{h}^k) || ||\Delta \textbf{h}^k ||\nonumber\\
&- L_{\max}\textstyle\sum_{i=1}^N || \Delta h_i^k ||^2.
\label{eq:Ineq_AfterLipschitz}
\end{align}
Since we defined the norm $||\Delta\textbf{h}||=(\textstyle\sum_{i=1}^N ||\Delta h_i||^2_\calH)^{1/2}$, we can rewrite the last expression as, 
\begin{align}
&F(\textbf{h}^{k+1})-F(\textbf{h}^{k})\nonumber\\
&\geq \alpha-||\hat{\nabla}_{\textbf{h}}F(\textbf{h}^k)|| ||\Delta \textbf{h}^k||-|| \nabla_{\textbf{h}} F(\textbf{h}^k) -  \hat{\nabla}_{\textbf{h}} F(\textbf{h}^k) || ||\Delta \textbf{h}^k ||\nonumber\\
&- L_{\max}||\Delta\textbf{h}^k||^2.
\label{eq:Ineq_Without_hi}
\end{align}
Using triangle inequality over the norm of cross-learning gradients $\hat{\nabla}_{\textbf{h}} F(\textbf{h}^k)$ function we can write, 
\begin{align}
F(&\textbf{h}^{k+1})-F(\textbf{h}^{k})\nonumber\\
&\geq \alpha-||\hat{\nabla}_{\textbf{h}}F(\textbf{h}^k)|| ||\Delta \textbf{h}^k||- L_{\max}||\Delta\textbf{h}^k||^2\nonumber\\
&-||\Delta \textbf{h}^k ||\textstyle\sum_{i=1}^N || \nabla_{h_i} U_i(h_i^k) -  \hat{\nabla}_{h_i} U_i(h_i^k) || .
\label{eq:Ineq_RightBeforeDeltas}
\end{align}
So far we only considered policies $\{h_i\}_{i=1}^N$ and not the central $g$.
Notice that because of how we defined the norm, expanding the policy vector with the cross-learning policy results in $||\textbf{z}||=||[\textbf{h},g]||\geq||\textbf{h}||$. By defining the expanded policy vector $\textbf{z}=[\textbf{h},g]$ we can rewrite \eqref{eq:Ineq_RightBeforeDeltas}, 
\begin{align}
F(&\textbf{h}^{k+1})-F(\textbf{h}^{k})\nonumber\\
&\geq \alpha-||\hat{\nabla}_{\textbf{h}}F(\textbf{h}^k)|| ||\Delta \textbf{z}^k||- L_{\max}||\Delta\textbf{z}^k||^2\nonumber\\
&-||\Delta \textbf{z}^k ||\textstyle\sum_{i=1}^N || \nabla_{h_i} U_i(h_i^k) -  \hat{\nabla}_{h_i} U_i(h_i^k) || .
\label{eq:Ineq_AddingG}
\end{align}

Where we have defined $\Delta \textbf{z}^k=\textbf{z}^{k+1}-\textbf{z}^{k}$ as the difference between the expanded policy vector $\textbf{z}^k=[\textbf{h}^k,g^k]$ and its next iteration, given by $\textbf{z}^{k+1}=\calP_{\calC}[\textbf{z}^k+ \eta^k \hat{\nabla}_{\textbf{z}} F(\textbf{z}^k)]+\textbf{b}^k=\calP_{\calC}[\{\textbf{h}^k+ \eta^k \hat{\nabla}_{\textbf{h}} F(\textbf{h}^k)\},g^k]+\textbf{b}^k$. The previous equality holds as the cross-learning function $F(\cdot)$ does not depend on the cross-learning policy $g$ and thus its derivative with respect to cross-learning policy is zero, $\nabla_{g} F(\cdot)=0$. As the projection to the convex set $\calP_\calC$ is non-expansive, using the triangle inequality we can write an inequality with respect to the stochastic gradient of  the expected discounted returns $U_i$,
\begin{align}
	||\Delta \textbf{z}^k||&=||\textbf{z}^{k+1}-\textbf{z}^k|| \nonumber\\
	&=||\calP_{\calC} [\textbf{z}^k +\eta^k \hat{\nabla}_{\textbf{z}} F(\textbf{z}^k) ]+\textbf{b}^k - \textbf{z}^k|| \nonumber \\
	&\leq ||\calP_{\calC} [\textbf{z}^k +\eta^k \hat{\nabla}_{\textbf{z}} F(\textbf{z}^k) ] - \textbf{z}^k||+||\textbf{b}^k|| \nonumber\\
	&\leq ||\eta^k \hat{\nabla}_{\textbf{z}} F(\textbf{z}^k)||+||\textbf{b}^k|| \nonumber\\
	&= ||\eta^k \hat{\nabla}_{\textbf{h}} F(\textbf{h}^k)||+||\textbf{b}^k|| \nonumber\\
	&\leq \eta^k|| \hat{\nabla}_{\textbf{h}} F(\textbf{h}^k)||+\textstyle\sum_{i=1}^{N+1}||b_i^k||\nonumber\\
	&\leq \eta^k\textstyle\sum_{i=1}^{N} ||\hat{\nabla}_{h_i} U_i(h_i^k)||+\sqrt{N+1}\beta. \label{eq:Ineq_Deltah_gradUi}
\end{align}
With a similar procedure, in the case of the square of the norm of the expanded policy vector, $||\Delta \textbf{z}^k||^2$ we can obtain an upper bound by bounding the inner product between the cross-terms as follows,
\begin{align}
	&||\Delta \textbf{z}^k||^2 =||\textbf{z}^{k+1}-\textbf{z}^k||^2 \nonumber\\
	&= ||\eta^k \hat{\nabla}_{\textbf{h}} F(\textbf{h}^k)||^2+(N+1)\beta^2+2\sqrt{N+1}\beta||\eta^k \hat{\nabla}_{\textbf{h}} F(\textbf{h}^k)|| \nonumber\\
	&= (||\eta^k \hat{\nabla}_{\textbf{h}} F(\textbf{h}^k)||+\sqrt{N+1}\beta)^2 \nonumber\\
	&\leq 3||\eta^k \hat{\nabla}_{\textbf{h}} F(\textbf{h}^k)||^2+3(N+1)\beta^2. \label{eq:Ineq_Deltah2_gradF}
\end{align}
Where we used $(a+b)^2\leq 3a^2+3b^2$. Hence, we can apply expressions \eqref{eq:Ineq_AddingG}, \eqref{eq:Ineq_Deltah_gradUi} and \eqref{eq:Ineq_Deltah2_gradF} to the bound in \eqref{eq:Ineq_RightBeforeDeltas}, to obtain the following inequality
\begin{align}
&F(\textbf{h}^{k+1})-F(\textbf{h}^{k})\nonumber\\
&\geq \alpha-||\hat{\nabla}_{\textbf{h}}F(\textbf{h}^k)|| (\eta^k|| \hat{\nabla}_{\textbf{h}} F(\textbf{h}^k)||+\sqrt{N+1}\beta) \nonumber\\
&- L_{\max}(3||\eta^k \hat{\nabla}_{\textbf{h}} F(\textbf{h}^k)||^2+3(N+1)\beta^2)\nonumber\\
&-\eta^k\textstyle\sum_{i=1}^N ||\hat{\nabla}_{h_i} U_i(h_i^k)||(\textstyle\sum_{i=1}^N || \nabla_{h_i} U_i(h_i^k) -  \hat{\nabla}_{h_i} U_i(h_i^k) ||)\nonumber\\
&-\sqrt{N+1}\beta(\textstyle\sum_{i=1}^N || \nabla_{h_i} U_i(h_i^k) -  \hat{\nabla}_{h_i} U_i(h_i^k)||)\\
&\geq \alpha-3(N+1)\beta^2L_{\max}- \sqrt{N+1}\beta||\hat{\nabla}_{\textbf{h}}F(\textbf{h}^k)|| \nonumber\\
&- (3 (\eta^k)^2 L_{\max}+\eta^k)|| \hat{\nabla}_{\textbf{h}} F(\textbf{h}^k)||^2\nonumber\\
&-\eta^k\textstyle\sum_{i=1}^N ||\hat{\nabla}_{h_i} U_i(h_i^k)||(\textstyle\sum_{i=1}^N || \nabla_{h_i} U_i(h_i^k) -  \hat{\nabla}_{h_i} U_i(h_i^k) ||)\nonumber\\
&-\sqrt{N+1}\beta(\textstyle\sum_{i=1}^N || \nabla_{h_i} U_i(h_i^k) -  \hat{\nabla}_{h_i} U_i(h_i^k)||).
\label{eq:Ineq_FirstAfterUsingbfU}
\end{align}
 Observe that the square of the norm of the stochastic gradient of the cross-learning function $F(\cdot)$ is equal to the sum of the squares of the term-wise norm, namely,
\begin{align}
||\hat{\nabla}_{\textbf{h}}F(\textbf{h}^k)||^2&=||\hat{\nabla}_{\textbf{h}}\textstyle\sum_{i=1}^NU_i(h_i^k)||^2\nonumber\\
&=||[\hat{\nabla}_{h_1}U_1(h_1^k),\dots,\hat{\nabla}_{h_N}U_N(h_N^k)]^T||^2\nonumber\\
&=\textstyle\sum_{i=1}^N ||\hat{\nabla}_{h_i}U_i(h_i^k)||^2. \label{eq:NormSquareSum}
\end{align}
Substituting \eqref{eq:NormSquareSum} for $||\hat{\nabla}_{\textbf{h}}F(\textbf{h}^k)||^2$ and using triangle inequality on the norm of the gradient of the cross-learning function in \eqref{eq:Ineq_FirstAfterUsingbfU} we obtain, 
\begin{align}
&F(\textbf{h}^{k+1})-F(\textbf{h}^{k})\nonumber\\
&\geq \alpha-3(N+1)\beta^2L_{\max}- \sqrt{N+1}\beta\textstyle\sum_{i=1}^N ||\hat{\nabla}_{h_i}U_i(h_i^k)|| \nonumber\\
&- (3 (\eta^k)^2 L_{\max}+\eta^k)(\textstyle\sum_{i=1}^N ||\hat{\nabla}_{h_i}U_i(h_i^k)||^2)\nonumber\\
&-\eta^k\textstyle\sum_{i=1}^N ||\hat{\nabla}_{h_i} U_i(h_i^k)||(\textstyle\sum_{i=1}^N || \nabla_{h_i} U_i(h_i^k) -  \hat{\nabla}_{h_i} U_i(h_i^k) ||)\nonumber\\
&-\sqrt{N+1}\beta(\textstyle\sum_{i=1}^N || \nabla_{h_i} U_i(h_i^k) -  \hat{\nabla}_{h_i} U_i(h_i^k)||).
\label{eq:Ineq_SecondAfterUsingbfU}
\end{align}
We can further add and subtract the deterministic gradient of the expected discounted returns $\nabla_{h_i} U_i(h_i^k)$ to the two terms that refer to the norm of the stochastic gradient of the expected discounted returns $||\hat{\nabla}_{h_i} U_i(h_i^k)||$ of equation \eqref{eq:Ineq_SecondAfterUsingbfU} and using the triangle inequality we obtain,
\begin{align}
&F(\textbf{h}^{k+1})-F(\textbf{h}^{k})\nonumber\\
&\geq \alpha-3(N+1)\beta^2L_{\max}\nonumber\\
&-\sqrt{N+1}\beta\textstyle\sum_{i=1}^N ||\hat{\nabla}_{h_i}U_i(h_i^k)-\nabla_{h_i}U_i(h_i^k)|| \nonumber\\
&- \sqrt{N+1}\beta\textstyle\sum_{i=1}^N ||\nabla_{h_i}U_i(h_i^k)|| \nonumber\\
&- (3 (\eta^k)^2 L_{\max}+\eta^k)(\textstyle\sum_{i=1}^N ||\hat{\nabla}_{h_i}U_i(h_i^k)||^2)\nonumber\\
&-\eta^k\textstyle\sum_{i=1}^N ||\nabla_{h_i} U_i(h_i^k)||(\textstyle\sum_{i=1}^N || \nabla_{h_i} U_i(h_i^k) -  \hat{\nabla}_{h_i} U_i(h_i^k) ||)\nonumber\\
&-\eta^k(\textstyle\sum_{i=1}^N || \nabla_{h_i} U_i(h_i^k) -  \hat{\nabla}_{h_i} U_i(h_i^k) ||)^2\nonumber\\
&-\sqrt{N+1}\beta(\textstyle\sum_{i=1}^N || \nabla_{h_i} U_i(h_i^k) -  \hat{\nabla}_{h_i} U_i(h_i^k)||).
\label{eq:Ineq_PreviousToExpectation}
\end{align}
Then, taking the expectation on both side and using Assumptions \ref{As:Max_Grad_Deterministic}, \ref{As:Variance_Grad} and \ref{As:DifferenceGradients} it yields
\begin{align}
&\EX[F(\textbf{h}^{k+1})-F(\textbf{h}^{k})]\nonumber \\
&\geq \alpha-3(N+1)\beta^2L_{\max}- \sqrt{N+1}\beta\textstyle\sum_{i=1}^N\frac{\sigma_{U_i}}{\sqrt{b_{U_i}}} \nonumber\\
&- \sqrt{N+1}\beta\textstyle\sum_{i=1}^N \mu_i \nonumber- (3 (\eta^k)^2 L_{\max}+\eta^k)(\textstyle\sum_{i=1}^N B_{U_i}^2)\nonumber\\
&-\eta^k\textstyle\sum_{i=1}^N \mu_i(\textstyle\sum_{i=1}^N \frac{\sigma_{U_i}}{\sqrt{b_{U_i}}})-\eta^k(\textstyle\sum_{i=1}^N \frac{\sigma_{U_i}}{\sqrt{b_{U_i}}})^2\nonumber\\
&-\sqrt{N+1}\beta(\textstyle\sum_{i=1}^N \frac{\sigma_{U_i}}{\sqrt{b_{U_i}}}).
\label{eq:Ineq_Equation_56}
\end{align}
Thus, by defining $\eta=\sup_k \eta_k$ and $\gamma_F$ to be
\begin{align}
&\gamma_F=\alpha-3(N+1)\beta^2L_{\max}- \sqrt{N+1}\beta\textstyle\sum_{i=1}^N\frac{\sigma_{U_i}}{\sqrt{b_{U_i}}} \nonumber\\
&- \sqrt{N+1}\beta\textstyle\sum_{i=1}^N \mu_i \nonumber- (3 \eta^2 L_{\max}+\eta)(\textstyle\sum_{i=1}^N B_{U_i}^2)\nonumber\\
&-\eta\textstyle\sum_{i=1}^N \mu_i(\textstyle\sum_{i=1}^N \frac{\sigma_{U_i}}{\sqrt{b_{U_i}}})-\eta(\textstyle\sum_{i=1}^N \frac{\sigma_{U_i}}{\sqrt{b_{U_i}}})^2\nonumber\\
&-\sqrt{N+1}\beta(\textstyle\sum_{i=1}^N \frac{\sigma_{U_i}}{\sqrt{b_{U_i}}}).
\end{align}
we obtain the desired result.
\end{proof}

\section{Proof of Lemma \ref{Lemma:UpperBoundIterations}} \label{app:Lemma:UpperBoundIterations}

\begin{proof}
	Given the iterates $\textbf{h}^K$ and the initial policies $\textbf{h}^0$, we can express the difference $\EX[F(\textbf{h}^K)-F(\textbf{h}^0)]$, as the summation over difference of iterates. That is,
	\begin{align}
	\EX[F(\textbf{h}^K)-F(\textbf{h}^0)]=\EX \bigg[\sum_{k=1}^K F(\textbf{h}^{k})-F(\textbf{h}^{k-1})\bigg].
	\end{align}
	Then, taking the expectation with respect to the final iterate at $t=K$, we have 
	\begin{align}
	\EX[F(\textbf{h}^K)-F(\textbf{h}^0)]=\EX_K\bigg[\EX\bigg[\sum_{k=1}^K F(\textbf{h}^{k})-F(\textbf{h}^{k-1})|K\bigg]\bigg]. \nonumber
	\end{align}
	Expanding the outermost expectation and applying the result of Lemma \ref{Lemma:F_Diff_Lower_Bound} it follows, 
	\begin{align}
	\EX[F(\textbf{h}^K)-F(\textbf{h}^0)]=& \sum_{t=0}^{\infty} \EX\bigg[\sum_{k=1}^t F(\textbf{h}^{k})-F(\textbf{h}^{k-1})\bigg] \Pr (K=t) \nonumber\\
	\geq&\gamma_F \sum_{t=0}^{\infty}t  \Pr (K=t) \nonumber\\
	\geq& \gamma_F \EX[K].
	\end{align}
	Rearranging the terms of the previous inequality gives us the desired result.
\end{proof}

\section{Proof of Theorem \ref{theo_HighProb}} \label{app:theo_HighProb}

\begin{proof} 	
	We begin with the inner product $\langle \nabla_{h_i} U_i (h_i^K), (h-h_i^K)\rangle$ and then add and subtract its stochastic counterpart $\langle \hat{\nabla}_{h_i} U_i (h_i^K), (h-h_i^K)\rangle$ to obtain,
	\begin{align}\label{eq:theo_sumsubs}
	\langle \nabla_{h_i} U_i &(h_i^K), (h-h_i^K)\rangle =\langle \hat{\nabla}_{h_i} U_i (h_i^K), (h-h_i^K)\rangle \nonumber\\
	&\quad +\langle(\nabla_{h_i} U_i (h_i^K)-\hat{\nabla}_{h_i} U_i (h_i^K)), (h-h_i^K)\rangle.
	\end{align}
	By construction, policies $h^K_i$ and any $h \in \calC$ must be $\epsilon$-close to the cross-learned policy $g$ hence $||h_i^k-h||\leq2\epsilon$, $\forall h \in \calC$, then using the stopping time condition \eqref{eq:stochasticalphafosp} and the Cauchy-Schwartz inequality, we can rewrite the previous expression as
	\begin{align}\label{eq:theo_firstineq}
	\langle \nabla_{h_i} U_i & (h_i^K), (h-h_i^K)\rangle   \nonumber\\
	&\leq \alpha+ 2 \epsilon  ||\hat{\nabla}_{h_i} U_i (h_i^K)-\nabla_{h_i} U_i(h_i^K)||\\ 
	&\leq \alpha+ 2 \epsilon \max_{i \in \{1,\ldots,N \} } ||\hat{\nabla}_{h_i} U_i (h_i^K)-\nabla_{h_i} U_i(h_i^K)||\nonumber.
	\end{align}
	Since the right hand side of \eqref{eq:theo_firstineq} does not depend on $h$ or $i$, then by computing the maximum over $h$ it holds that
	\begin{align}\label{eq:theo_secondineq}
	&\max_{ \substack {h \in \calC}} \langle \nabla_{h_i} U_i (h_i^K), (h-h_i^K)\rangle \nonumber\\
	&\quad \leq \alpha + 2 \epsilon \max_{i \in \{1,\ldots,N \} } ||\hat{\nabla}_{h_i} U_i (h_i^K)-\nabla_{h_i} U_i(h_i^K)||. 
	\end{align}
	Using Markov's inequality together with Assumption \ref{As:DifferenceGradients}, it yields,
	\begin{align}\label{theo:chebyshev}
	\Pr\biggl(\max_{i \in \{1,\ldots,N \} }||\hat{\nabla}_{h_i} U_i & (h_i^K)-\nabla_{h_i} U_i(h_i^K)||\leq\xi\biggr) \nonumber\\
	&\quad \geq 1-\frac{1}{\xi^2}\max_{i \in \{1,\ldots,N \} } \frac{\sigma_{U_i}^2}{b_{U_i}}.
	\end{align}
	Hence, combining \eqref{eq:theo_secondineq} and \eqref{theo:chebyshev} it follows that
	\begin{align}\label{eq:final_eq}
	\Pr\biggl( \max_{i \in \{1,\ldots,N \}} \langle \nabla_{h_i}  U_i (h_i^K), & (h-h_i^K)\rangle \leq \alpha + 2\epsilon\xi \biggr) \nonumber\\
	&\geq 1-\frac{1}{\xi^2}\max_{i \in \{1,\ldots,N \} } \frac{\sigma_{U_i}^2}{b_{U_i}}.
	\end{align}
	Then we can define the following constants, $\alpha'=\alpha + 2\epsilon \xi$ and $\delta=\frac{1}{\xi^2}\max_{i \in \{1,\ldots,N \} } \frac{\sigma_{U_i}^2}{b_{U_i}}$ to obtain
	\begin{align}
	\Pr\left( \max_{i \in \{1,\ldots,N \} } \langle \nabla_{h_i} U_i (h_i^{K}), (h-h_i^{K}) \rangle \leq \alpha' \right) 
	\geq 1-\delta
	\end{align}	
	where we can also rearrange terms to obtain the value $\alpha'=\alpha + 2\epsilon\max_{i \in \{1,\ldots,N \}} \sigma_{U_i} \sqrt{\frac{1}{\delta b_{U_i}}}$, which is the intended result. 
\end{proof}